\documentclass[a4paper,UKenglish,cleveref, autoref,thm-restate]{lipics-v2021}
\usepackage{amsmath}
\usepackage{amsfonts}
\usepackage{amssymb}
\usepackage{graphicx}
\usepackage{tikz}
\usepackage{esvect}
\usepackage[ruled,linesnumbered,longend]{algorithm2e}
\SetArgSty{textnormal}  
\usepackage{caption}
\usepackage{subcaption}
\usepackage{amsthm}
\usepackage{cleveref}
\usepackage{comment}
\usepackage{lineno}
\usepackage{comment}

\usetikzlibrary{arrows,calc,shapes.misc}
\tikzstyle{vertex}=[scale=0.8,draw,circle,fill=white,draw=black]
\tikzstyle{blocked}=[line width=3pt,red]
\tikzstyle{tblocked}=[line width=1.5pt,red]
\tikzstyle{gblocked}=[line width=2pt,white!50!gray]
\tikzstyle{free}=[line width=4pt,black]
\tikzstyle{thin}=[line width=2pt,black]
\tikzstyle{terminal}=[scale=0.8,circle,fill=orange]
\tikzstyle{rterminal}=[scale=0.8,circle,fill=red]
\tikzstyle{redarrow}=[line width=3pt,red,->]
\tikzstyle{dredarrow}=[line width=3pt,red,<->]
\tikzstyle{bluearrow}=[line width=3pt,blue,->]
\tikzstyle{dbluearrow}=[line width=3pt,blue,<->]
\tikzstyle{greenarrow}=[line width=3pt,black!30!green,->]
\tikzstyle{dgreenarrow}=[line width=3pt,black!30!green,<->]

\newtheorem*{mainresult}{Main result}

\newcommand{\kctp}{$k$-CTP\xspace}

\newcommand{\dopt}{d_{\mbox{\scriptsize{opt}}}}

\newcommand{\dopte}{d_{\emph{\scriptsize{opt}}}}
\newcommand{\popte}{P_{\emph{\scriptsize{opt}}}}

\newcommand{\dtr}[1]{d_{ #1 }^{\mbox{\scriptsize{Tr}}}}
\newcommand{\rb}{\emph{RB}}
\newcommand{\rbm}{\mbox{RB}}

\newcommand{\ve}[1]{\mbox{\bf #1}}

\newcommand{\flw}{\textsc{FollowWall}}
\newcommand{\dflw}{\textsc{Doubling-FollowWall}}

\title{Tight asymptotic competitive ratio for the Canadian Traveller Problem on unit-weighted grids}

\titlerunning{Competitive analysis for the \kctp\ on unit-weighted grids}

\author{Pierre {Berg\'e}}{Univ. Grenoble Alpes, CNRS, Grenoble INP, LIG, Grenoble, 38000, France}{pierre.berge@univ-grenoble-alpes.fr}{0000-0002-6170-9473}{}

\author{Vsevolod {Chernyshev}}{Ulm University, Ulm, Germany}{vchern@gmail.com}{}{}

\author{Malory {Marin}}{Univ. Grenoble Alpes, CNRS, Grenoble INP, LIG, Grenoble, 38000, France}{malory.marin@univ-grenoble-alpes.fr}{0009-0008-8253-2831}{}

\author{Alexey {Pyrko}}{Higher School of Economics, Moscow, Russia}{}{}{}

\author{Lou {Sala\"un}}{Nokia Bell Labs, Massy, 91300, France}{lou.salaun@nokia-bell-labs.com}{0000-0003-0522-0435}{}

\authorrunning{P. {Berg\'e}, V. {Chernyshev}, M. {Marin}, A. {Pyrko} and L. {Sala\"un}} 

\Copyright{Pierre Berg\'e, Vsevolod Chernyshev, Malory Marin, Alexey Pyrko, and Lou Sala\"un}



\category{} 

\relatedversion{} 

\nolinenumbers 

\hideLIPIcs

\EventEditors{John Q. Open and Joan R. Access}
\EventNoEds{2}
\EventLongTitle{42nd Conference on Very Important Topics (CVIT 2016)}
\EventShortTitle{CVIT 2016}
\EventAcronym{CVIT}
\EventYear{2016}
\EventDate{December 24--27, 2016}
\EventLocation{Little Whinging, United Kingdom}
\EventLogo{}
\SeriesVolume{42}
\ArticleNo{23}

\begin{document}

\maketitle

\begin{abstract}
We study the $k$-\textit{Canadian Traveller Problem} (\kctp), an online navigation problem introduced by Papadimitriou and Yannakakis in 1991. Given an undirected weighted graph $G=(V,E,\omega)$ and two vertices $s,t\in V$, a traveller aims at navigating from $s$ to $t$ while minimizing the total distance travelled. However, up to $k$ edges of the graph may be blocked, and the traveller discovers that an edge is blocked only upon reaching one of its endpoints. The performance of an online algorithm for \kctp\ is measured with the competitive ratio, which is the quotient between the distance actually traversed and the distance traversed by an omniscient traveller who knows all blockages in advance.

A general lower bound $2k+1$ for the competitive ratio of deterministic algorithms remains valid even when the input graph is series-parallel, whereas constant competitive ratios can be achieved on restricted graph classes, such as cycles or unit-weighted outerplanar graphs. In this paper, we establish the first intermediate competitive ratio between constant and linear regimes. We prove that, on unit-weighted solid grid graphs, there exists an online algorithm with competitive ratio $O(\sqrt{k})$. This ratio is asymptotically tight and cannot be defeated, even on unit-weighted square grids. Finally, we show that both the unit-weighted and solidity assumptions are essential: going further with arbitrary edge weights (even on square grids) or non-solid unit-weighted grid structures yields a linear lower bound of $\Omega(k)$.
\end{abstract}

\newpage


\section{Introduction}

We study an online navigation problem on graphs, the $k$-\textit{Canadian Traveller Problem} (\kctp), introduced by Papadimitriou and Yannakakis~\cite{PaYa91}. Given an undirected connected weighted graph $G=\left(V,E,\omega\right)$ and a source-target pair $s,t \in V$, a traveller must navigate from $s$ to $t$ in $G$ while minimizing the total distance travelled. Edge weights are positive. Unfortunately, certain edges (at most $k$) of the graph can be blocked. 
The traveller discovers that an edge $e=(u,v)$ is blocked upon reaching one of its endpoints $u$ or $v$. The \kctp is PSPACE-complete~\cite{BaSc91,PaYa91}. Several variants have been studied, including settings in which edges are blocked with a certain probability~\cite{AkSaAr16,BaSc91,fried2013complexity}, predictions are given~\cite{BaEsXe22}, remote edges can be sensed at a cost~\cite{fried2013complexity}, or the objective is to find a shortest tour~\cite{HaXe23,LiHu14}. This problem has applications in robot routing for humanitarian logistics~\cite{AkSaAr16,AlYiAk21,EyKeHe10}.

Any instance of the \kctp\ can be divided into two parts: a visible one, the weighted graph $G=\left(V,E,\omega\right)$ with $s,t \in V$ and a hidden one, the set of blocked edges $E_* \subsetneq E$, $\vert E_* \vert \le k$.
The objective is to make the traveller attain the target $t$ at a minimum cost, {\em i.e.} the sum of the weights of the traversed edges. An instance $\left(G,E_*\right)$ is called a \textit{road map}.
Road maps are assumed \emph{feasible}: $s$ and $t$ belong to the same connected component of the graph obtained from $G$ by removing $E_*$. In other words, there is always a way to complete the walk despite the blockages.

A solution to the \kctp is an online algorithm, called a \textit{strategy}, which guides the traveller along the walk from $s$ to $t$ on the graph. The \emph{competitive ratio}~\cite{BoEl98} of a strategy on a given instance is the quotient between the distance actually traversed by the traveller on that instance and the distance they would have traversed, knowing which edges are blocked in advance. Westphal~\cite{We08} proved that, for any deterministic strategy $A$, there is at least one \kctp road map for which the competitive ratio of $A$ is at least $2k+1$.

One way to circumvent this lower bound consists in adding randomization~\cite{BeWe15,DeHuLiSa14}. A second approach, which we pursue in this paper, is to restrict the competitive analysis of deterministic strategies to certain families of graphs. The $2k+1$ lower bound remains valid even when restricted to unit-weighted series-parallel\footnote{a sub-family of planar graphs of treewidth at most $2$} graphs~\cite{ChChWuWu15}. However, positive results are known for some families: competitive ratio $3$ for cycles and cacti (folklore), competitive ratio $9$ for unit-weighted outerplanar graphs~\cite{BeBeChDa24} and competitive ratio $\sqrt{2}\cdot k+O(1)$ for graphs where the maximum $(s,t)$-cut has a bounded size~\cite{BeSa23}. 

In this paper, we study the \kctp\ on unit-weighted solid grids. Informally, these are induced subgraphs of the square grid whose bounded faces are all unit squares. In particular, this family contains square and rectangular grids. Besides their simple geometric structure, solid grids provide a natural abstraction for studying navigation in transportation networks.

Xu {\em et al.}~\cite{XuHuSuZh09} previously investigated the design of efficient strategies for unit-weighted rectangular grids. They showed that when $k$ is smaller than both the horizontal and vertical widths of the shortest $(s,t)$-path, then the \textsc{Greedy} strategy achieves a constant competitive ratio $3$. Strategy \textsc{Greedy} simply consists, at each moment of the walk, of going through the shortest path between the current position to target $t$, taking into account the already discovered blockages. Their assumption on $k$, however, substantially simplifies the set of instances under consideration. Indeed, in this paper, we prove that the competitive ratio of \textsc{Greedy} without any restriction on $k$ is at least linear in $k$, even on square grids.

We thus propose a new strategy, called \dflw\, achieving a smaller competitive ratio on all unit-weighted solid grid graphs. Our main contribution is summarized:

\begin{mainresult}
There is a polynomial-time deterministic strategy which achieves competitive ratio $O(\sqrt{k})$ for the \kctp\ on unit-weighted solid grids. This ratio is asymptotically tight, since no deterministic strategy can achieve a competitive ratio 
$o(\sqrt{k})$ on this family.
\end{mainresult}

Some techniques employed in our work bear similarities with the ones proposed for outerplanar graphs in~\cite{BeBeChDa24}. Indeed, a key limitation of \textsc{Greedy} strategy is its tendency to persist in one direction when successive blockages are discovered. The distance traversed thus depends linearly on $k$, whereas a constant distance could have been covered on some other open $(s,t)$-path. To address this shortcoming, we introduce the \dflw\ strategy, which always keep a close look on the possible open $(s,t)$-distance instead of insisting on one direction full of blockages. To do so, we assign to the traveller a \emph{distance budget}, which doubles at each iteration: a similar technique was used in~\cite{BeBeChDa24}. In addition, several topological arguments compatible with the grid embedding are used in our competitive analysis for \dflw , leading to ratio $O(\sqrt{k})$. To the best of our knowledge, it is the first time that such tools are employed for the study of \kctp .

A natural question would be whether our result can be extended to grids without weight restrictions. Unfortunately, we show that it is not possible to outperform linear-in-$k$ competitive ratio, even on arbitrarily-weighted square grids.

In Section~\ref{sec:prelim}, we introduce graph notations, the definition of the \kctp , and some preliminary observations on grids. 
Section~\ref{sec:insight} aims at establishing lower bounds to motivate our main contribution: we derive a lower bound $\Omega(\sqrt{k})$ for the competitive ratio of unit-weighted square grids, and prove that \textsc{Greedy} is far from this bound.
In Section~\ref{sec:strategy}, we introduce our strategy \flw\ and show that it achieves a competitive ratio $O(\sqrt{k})$ when the size of the solid grid graph is small relative to the $(s,t)$-distance. This serves as a preliminary result for Section~\ref{sec-double} which generalizes it to all unit-weighted solid grids by introducing the \dflw\ strategy. For all results annotated with a $(\star)$, the full proof is given in the Appendix.

\section{Preliminaries} \label{sec:prelim}

\subsection{Graph notation}

We consider only undirected connected weighted graphs $G=(V,E,\omega)$, where $\omega : E \rightarrow \mathbb{Q}^+$. 
A graph is \textit{unit-weighted} if $\omega(e) = 1$ for every edge $e\in E$. 
A subgraph $G'$ of $G$ is a graph $G'=(V',E',\omega')$, where $V' \subseteq V$, $E' \subseteq E\cap (V' \times V')$, and $\omega' = \omega_{|E'}$, {\em i.e.} $\omega'$ restricted to set $E'$. 
For any vertex set $V' \subseteq V$, we denote by $E[V'] = E \cap (V' \times V')$ the edge set induced by $V'$ and $G\left[V'\right]$ the subgraph of $G$ induced by $V'$: $G\left[V'\right] = \left(V',E[V'],\omega_{|E[V']}\right)$. Let $G\setminus V'$ be the graph deprived of set $V'$: $G\setminus V' = G[V \setminus V']$. For any $V'\subseteq V(G)$, we denote by $N_G[V']$ the closed neighborhood of $V'$, consisting of all vertices in $V'$ and all vertices adjacent to at least one vertex of $V'$. Similarly, for any edge set $E' \subseteq E$, we denote by $G\left[E'\right]$ the subgraph of $G$ induced by $E'$: $G\left[E'\right] = \left(V,E',\omega_{|E'}\right)$. Let $G\setminus E'$ be the graph deprived of edges $E' \subseteq E$: $G\setminus E' = \left(V,E\setminus E',\omega_{|E \setminus E'}\right)$. 


A \textit{simple path} (called simply a \emph{path} in the remainder) $P$ of $G$ is a sequence of pairwise different vertices $v_1\cdot v_2\cdots v_i\cdot v_{i+1} \cdots v_{\ell}$, with departure $v_1$ and arrival $v_{\ell}$, where $v_iv_{i+1} \in E$ for $1\le i\le \ell-1$. A path $P$ between vertices $u$ and $v$ is called a \emph{$(u,v)$-path}. We abuse notation: $v_1 \in P$ and $v_1v_2 \in P$ mean that vertex $v_1$ and edge $v_1v_2$ are on the path $P$, respectively. We denote by $P[v_i,v_j]$ the sub-path of $P$ obtained from picking the section of $P$ between $v_i$ and $v_j$.
In a $(u,v)$-\textit{walk}, however, vertices can be repeated. Moreover, an edge can be traversed several times in a walk, and in opposite directions. 

A \emph{simple cycle} $C$ of $G$ is a sequence of pairwise different vertices $v_1\cdot v_2 \cdots v_{\ell}\cdot v_1$ such that $v_iv_{i+1}\in E$ for every $1\le i \le \ell-1$, and $v_\ell v_1\in E$ as well.


Throughout the paper, we use classical Landau notation $O$, $o$, and $\Omega$.

\subsection{Problem definition and competitive analysis}

Let $G=(V,E,\omega)$ be a graph and $E_*\subsetneq E$ a set of blocked edges. A \textit{road map}, which is a feasible instance of the \kctp problem, is a 4-tuple $(G,s,t,E_*)$ such that $s$ and $t$ are connected in $G\setminus E_*$.
We can now formally introduce the $k$-Canadian Traveller Problem.

\begin{definition}[$k$-Canadian Traveller Problem, or \kctp]~
	
	\textbf{\emph{Input:}} A graph $G=(V,E,\omega)$, $s,t \in V$, a set $E_*$ of blocked edges whose identities are unknown and such that $|E_*|\le k$ and $(G,s,t,E_*)$ is a road map.
	
	\textbf{\emph{Objective:}} Traverse graph $G$ from $s$ to $t$ while minimizing the distance travelled.
\end{definition}

Guided by a strategy, the traveller produces an $(s,t)$-walk.
The set of blocked edges $E_*$ is a hidden input at the beginning of the walk. We say an edge is \textit{revealed} when one of its endpoints has already been visited. A \textit{discovered blocked edge} is a revealed edge which is blocked.
At any moment of the walk, we usually denote by $E_*'$ the set of discovered blocked edges. 
At any moment, we are in fact working on the graph $G\setminus E_*'$ as discovered blocked edges can be permanently withdrawn from the input graph $G$. Conversely, an \emph{open edge} is an edge which is not blocked. We say a path is open if it does not contain any blocked edges.


For any subset of blocked edges $F\subseteq E_*$ and two vertices $x,y$, let $d_F\left(x,y\right)$ be the cost of the shortest $(x,y)$-path in the graph $G\setminus F$. Independently from blocked edges, we also write $d_{\infty}(x,y)$ the infinite-norm distance (or $\ell_\infty$-distance) between $x$ and $y$, {\em i.e.} the maximum gap between either horizontal or vertical coordinates of both points. Hence, for a given vertex $x \in V$, the set of vertices $y$ at some fixed infinite-norm distance from $x$ form a square cycle in the grid embedding.
In order to introduce competitive analysis in the \kctp\ context, we define the optimal offline cost, which is the distance the traveller would have traversed knowing the blockages in advance.

\begin{definition}[Optimal offline path and cost]
    An \emph{optimal offline path}, denoted by $\popte$, is a shortest $(s,t)$-path in the graph $G\setminus E_*$.
    Its cost, the \emph{optimal offline cost}, is $\dopte = d_{E_*}\left(s,t\right)$.
    \label{def:offline}
\end{definition}

The \textit{competitive ratio} is defined in \cite{BoEl98}. We denote by $\dtr{A}\left(G,s,t,E_*\right)$ the distance traversed by the traveller guided by a strategy $A$ on the graph $G$ from the source $s$ to the target $t$ with blocked edges $E_*$. Since $s$ and $t$ are generally clear from the context, we might simply write $\dtr{A}\left(G,E_*\right)$. The same remark holds for the notation of competitive ratio, presented below.

\begin{definition}[Competitive ratio]
The competitive ratio $c_A(G,E_*)$ of $A$ over a road map $(G,s,t,E_*)$ is the ratio between $d_{A}^{\emph{\scriptsize{Tr}}}\left(G,E_*\right)$ and the optimal offline cost. Formally:
$$c_A(G,E_*) = \frac{d_{A}^{\emph{\scriptsize{Tr}}}\left(G,E_*\right)}{d_{\emph{\scriptsize{opt}}}}$$
\end{definition}

We can assume w.l.o.g.\ that all blocked edges $E_*$ of our road map $(G,s,t,E_*)$ are accessible from $s$ in $G$: formally, for any $(u,v) \in E_*$, there exists an $(s,u)$-path or an $(s,v)$-path in $G\setminus E_*$. Indeed, all blocked edges which are not accessible from $s$ play no role in the behaviour of the traveller since they cannot be discovered. Said differently, as we conduct here a worst-case competitive analysis, useless blocked edges do not appear in the worst instances of \kctp .


\subsection{Grids}

\subparagraph{Square grids, rectangular grids and solid grid graphs.}
Let $\mathcal{G}_{\mathbb{Z}^2}$ denote the infinite square grid graph, whose vertex set is $\mathbb{Z}^2$ and in which two vertices are adjacent if and only if their $\ell_1$-distance is one. A \emph{grid graph} is a finite induced subgraph of $\mathcal{G}_{\mathbb{Z}^2}$. 
We consider the natural planar embedding of a grid graph inherited from $\mathbb{Z}^2$. 

We distinguish two classes of grid graphs that will be of particular interest: rectangular grids and solid grid graphs (Figure~\ref{fig:SolidGrid}).

\begin{definition}[Rectangular grid]
A \emph{rectangular grid} is a grid graph induced by a vertex set of the form $\{0,\ldots,n\}\times\{0,\ldots,p\}$ for some $n,p\in\mathbb{N}$.
\end{definition}

Equivalently, a rectangular grid is the Cartesian product of two finite paths. Thus, the boundary of a rectangular grid forms a rectangle and all its bounded faces are unit squares. When $n=p$, the graph is a \emph{square grid}.
Solid grid graphs generalize rectangular grids by allowing the boundary to have an arbitrary shape while still excluding non-unit bounded faces.

\begin{definition}[Solid grid graph~\cite{UmLe97}]
A \emph{solid grid graph} (or simply \emph{solid grid}) is a grid graph in which every bounded face has unit area.
\end{definition}

Hence, every square grid is a rectangular grid, and each rectangular grid is a solid grid graph. Figure~\ref{subfig:SolidGrid} illustrates a solid grid graph whose boundary is not rectangular. 

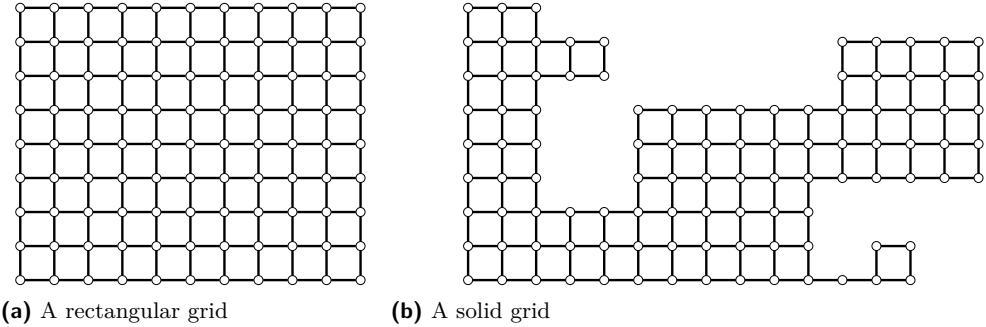
\begin{figure}[ht]
\centering
\begin{subfigure}[b]{0.36\columnwidth}
\centering
\scalebox{0.45}{\begin{tikzpicture}

\def\n{11} 
\def\m{9} 

\foreach \I in {1,...,\n} {
    \foreach \J in {1,...,\m} {
        \node[vertex] (v-\I-\J) at (\I,\J) {};
    }
}

\foreach \I in {2,...,\n} {
    \foreach \J in {1,...,\m} {
        \pgfmathtruncatemacro{\Imo}{\I-1} 
        \draw[thin] (v-\I-\J) to (v-\Imo-\J);
    }
}

\foreach \I in {1,...,\n} {
    \foreach \J in {2,...,\m} {
        \pgfmathtruncatemacro{\Jmo}{\J-1} 
        \draw[thin] (v-\I-\J) to (v-\I-\Jmo);
    }
}


\end{tikzpicture}}
\caption{A rectangular grid}
\end{subfigure}
\begin{subfigure}[b]{0.63\columnwidth}
\centering
\scalebox{0.45}{\begin{tikzpicture}
\foreach \i in {0,...,5}{
	\draw[thin] (\i,0) -- (\i, 5);
	\draw[thin] (0,\i) -- (5,\i);
}

\foreach \i in {0,...,4}{
	\draw[thin] (6, \i+3) -- (10, \i+3);
	\draw[thin] (\i+6, 3) -- (\i+6, 7);
}

\foreach \j in {3,4,5}{
	\draw[thin] (5,\j) -- (6,\j);
}

\foreach \i in {0,1,2}{
	\draw[thin] (-5, \i) -- (0, \i);
}

\foreach \i in {3,4,5, 8}{
	\draw[thin] (-5, \i) -- (-3, \i);
}

\foreach \i in {6,7}{
	\draw[thin] (-5, \i) -- (-1, \i);
}

\foreach \i in {-5,-4,-3}{
	\draw[thin] (\i, 0) -- (\i, 8);
}

\foreach \i in {-2,-1}{
	\draw[thin] (\i, 0) -- (\i, 2);
	\draw[thin] (\i, 6) -- (\i, 7);
}

\draw[thin] (5, 0) -- (8,0);
\draw[thin] (7, 1) -- (8,1);
\draw[thin] (7, 0) -- (7,1);
\draw[thin] (8, 0) -- (8,1);


\foreach \i in {0,...,5}{
	\foreach \j in {0,...,5}{
        \node[vertex] (v-\i-\j) at (\i,\j) {};
	}
}

\foreach \i in {0,...,4}{
	\foreach \j in {0,...,4}{
        \node[vertex] (u-\i-\j) at (\i+6,\j+3) {};
	}
}

\foreach \i in {0,...,4}{
	\foreach \j in {0,...,2}{
        \node[vertex] (w-\i-\j) at (-1-\i,\j) {};
	}
}

\foreach \i in {2,3,4}{
	\foreach \j in {0,...,5}{
        \node[vertex] (t-\i-\j) at (-1-\i,\j+3) {};
	}
}

\foreach \i in {-1,-2}{
	\foreach \j in {6,7}{
        \node[vertex] (t-\i-\j) at (\i,\j) {};
	}
}

\node [vertex] (a) at (6, 0) {};

\foreach \i in {7,8}{
	\foreach \j in {0,1}{
        \node[vertex] (a-\i-\j) at (\i,\j) {};
	}
}
\end{tikzpicture}}
\caption{A solid grid}
\label{subfig:SolidGrid}
\end{subfigure}
\caption{Examples of grid graphs}
\label{fig:SolidGrid}
\end{figure}

\subparagraph{The case of weighted square grids and unit-weighted grid graphs.}

Since our goal is to assess the competitiveness of deterministic strategies on different types of grids, we begin by providing two cases which are (almost) equivalent to the worst-case instances of \kctp, as presented in~\cite{We08}. Concretely, any strategy dealing with some graph made up of $k+1$ vertex-disjoint $(s,t)$-paths cannot outperform the competitive ratio $2k+1$. In fact, such a vertex-disjoint path structure can be built either by assignning arbitrary weights on any sufficiently large grid, or by considering unit-weighted grid graphs with large holes.

\begin{theorem}[$\star$]
There is no deterministic strategy for the \kctp\ with competitive ratio strictly smaller than $2k-1$ on unit-weighted grid graphs. The same holds for arbitrarily-weighted solid grid graphs.
\label{thm:lower-bound-weight}
\end{theorem}

The consequence of this theorem is that, if we aim at achieving a competitive ratio $o(k)$ on certain grid families, then we have to consider either weight or topological restrictions. For this reason, in the remainder, we focus on unit-weighted solid grids.

\section{First insights for the study of unit-weighted solid grids} \label{sec:insight}

We present here some initial observations on the study \kctp\ grid instances. First, we assess the competitiveness of the \textsc{greedy} strategy, which was first introduced in~\cite{XuHuSuZh09}: in the context of unit-weighted square grids, we prove that its competitive ratio is asymptotically linear in $k$. Second, we provide a general asymptotic lower bound $\Omega(\sqrt{k})$ for the competitive ratio on unit-weighted square grids. 

\subsection{\textsc{Greedy} strategy}

The \textsc{Greedy} strategy is a natural algorithm proposed by Xu {\em et al.}~\cite{XuHuSuZh09}. In this paper, the authors provide several results and observations on the competitiveness of this strategy. Its principle is the following: at each step, follow the shortest $(p,t)$-path in the graph $G \setminus E_*'$, where $p$ is the current position of the traveller and $E_*'$ is the set of blocked edges already discovered by the traveller. \textsc{Greedy} is the strategy we could apply in a real-world transportation network, when we encounter a blocked road (accident, road renovations,\ldots) and we aim at computing a short detour to reach our destination. 

Nevertheless, on general graphs, its competitive ratio can be exponentially large in $k$~\cite{XuHuSuZh09}, which is very far from the optimal ratio $2k+1$, achieved by the other strategy \textsc{reposition}~\cite{We08}. One could think that \textsc{Greedy} is efficient for some network topologies, in particular grid-like networks, which are natural models for urban transportation. In~\cite{XuHuSuZh09}, the authors show that, on unit-weighted rectangular grids, it achieves competitive ratio $3$ under a very specific and strong restriction: the number of blockages is upper-bounded by $d_{\infty}(s,t)$.

We prove that this statement is no longer valid when $k$ increases, especially when $k$ becomes much larger than the initial $(s,t)$-distance. In this scenario, the competitive ratio of the \textsc{Greedy} is at least linear in $k$, which does not outperform asymptotically the optimal bound $2k+1$ of general graphs. Therefore, \textsc{Greedy} does not offer meaningful guarantees on grids when the number of blockages is unbounded.

\begin{theorem}
The competitive ratio of \textsc{greedy} strategy on unit-weighted square grids is at least $\frac{2}{5}k + \frac{1}{5}$.
\end{theorem}
\begin{proof}
Consider a square grid of dimension $N \times N$, with $N$ odd and $N \ge 2k-1$. Place $s$ and $t$ at distance $3$ on the same horizontal line, such that the center of the grid is on the $3$-edge shortest $(s,t)$-path $P$. We fix $P = (s,u_0,v_0,t)$ and $u^-$ the vertex located below $u_0$.

Following \textsc{Greedy}, the traveller traverses the edge $(s,u_0)$. Assume it discovers that both $(u_0,v_0)$ and $(u_0,u^-)$ are blocked edges. Hence, we recompute the shortest $(u_0,t)$-path in the graph $G \setminus E_*'$, with $E_*' = \{(u_0,v_0),(u_0,u^-)\}$. This path goes up and the traveller visits a first vertex $u_1$, just above $u_0$. Assume the horizontal edge $(u_1,v_1)$, which is parallel to $(u_0,v_0)$, is blocked. Again, we recompute the current shortest path which goes up. The traveller moves to $u_2$ but discovers that $(u_2,v_2)$ is blocked, and so on. With the \textsc{Greedy} strategy, this process continues until some edge $(u_{k-2},v_{k-2})$ is blocked. Then the traveller finds an open detour via $(u_{k-1},v_{k-1})$ (which exists as $N \ge 2k-1$) and reaches $t$ without encountering extra blockages. The journey of the traveller is highlighted with green arcs in Figure~\ref{subfig:greedy}.

In this situation, the distance traversed by the traveller is exactly $2(k-1)+3$, while $\dopt = 5$. So, the competitive ratio of \textsc{greedy} on square grids is at least $\frac{2}{5}k + \frac{1}{5}$.
\end{proof}

\subsection{Lower bound for the competitiveness on unit-weighted grids} \label{subsec:lower}

This section is dedicated to proving that no competitive ratio $o(\sqrt{k})$ can be expected on unit-weighted square grids. To establish this result, we consider a square grid with certain configurations of blocked edges.

\begin{theorem}
There is no deterministic strategy that achieves a competitive ratio $o(\sqrt{k})$ on unit-weighted square grids for the \kctp .
\end{theorem}
\begin{proof}

Let us consider all road maps $(G,s,t,E_*)$, where $G$ is a $(N+1) \times (N+1)$ square grid, $s$ is the top-left corner vertex, $t$ is the bottom-right corner vertex, and $E_*$ is a set of blocked edges of size at most $k$. We fix the source and target to be at coordinates $\ve{x}(s) = 0$ and $\ve{y}(s) = N$; $\ve{x}(t) = N$ and $\ve{y}(t) = 0$. We restrict ourselves to blockage sets $E_*$ containing:
\begin{itemize}
    \item all vertical edges, except the ones present on the leftmost vertical line (the $\ve{x}$-coordinate is equal to 0) and on the rightmost one (the $\ve{x}$-coordinate is equal to $N$)
    \item all horizontal edges adjacent to the rightmost vertical line, except an arbitrary one.
\end{itemize}

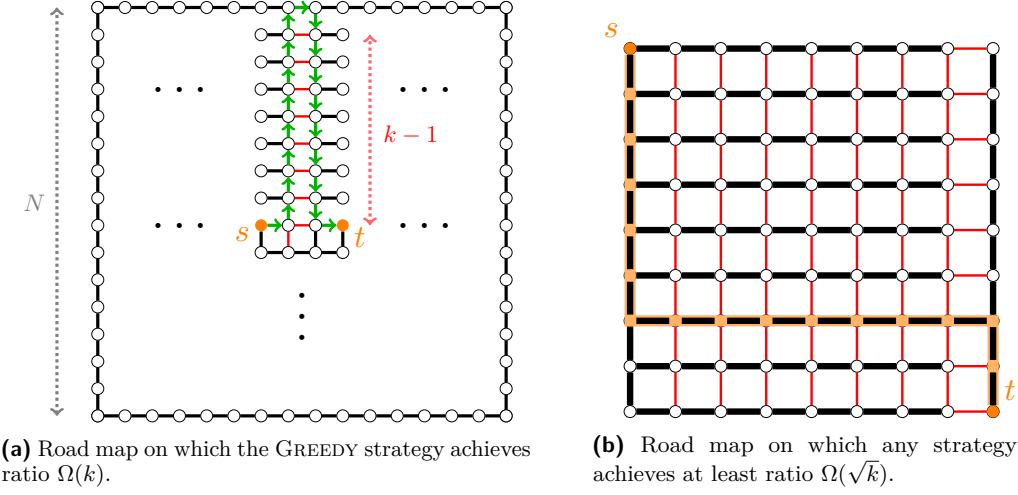
\begin{figure}[t]
	\centering
	\begin{subfigure}[b]{0.5\columnwidth}
	\centering
	\scalebox{0.6}{\begin{tikzpicture}

\def\l{15} 
\pgfmathtruncatemacro{\lmo}{\l - 1};
\pgfmathsetmacro{\H}{\lmo / 2};
\pgfmathsetmacro{\B}{6 - 0.6 * \H};
\pgfmathsetmacro{\T}{\B + 0.6 * \l};
\pgfmathsetmacro{\Tm}{\T - 0.6};

\node[terminal] (s) at (6,6) {};
\draw (s) node[below left,yshift = 2mm, scale=2,color = orange] {$s$};

\node[terminal] (t) at (7.8,6) {};
\draw (t) node[below right,yshift = 2mm,scale=2,color = orange] {$t$};

\draw[<->, dotted, line width = 2pt, color = gray] (1.5,\B) to (1.5,\T);
\node[color = gray, scale = 1.5] at (1,6.5) {$N$};
\draw[<->, dotted, line width = 2pt, color = white!40!red] (8.4,6) to (8.4,\Tm);
\node[color = red, scale = 1.5] at (9.3,8) {$k-1$};

\foreach \I in {0,...,\H} {
	\pgfmathsetmacro{\J}{6 + (\I * 0.6)};
   \node[vertex] (u-\I) at (6.6,\J) {};
   \node[vertex] (v-\I) at (7.2,\J) {};
   \draw[tblocked] (u-\I) to (v-\I);
}

\foreach \I in {1,...,\H} {
	\pgfmathsetmacro{\J}{6 + (\I * 0.6)};
   \node[vertex] (s-\I) at (6,\J) {};
   \node[vertex] (t-\I) at (7.8,\J) {};
   \draw[thin] (s-\I) to (u-\I);
   \draw[thin] (v-\I) to (t-\I);
}

\foreach \I in {1,...,\H} {
	\pgfmathtruncatemacro{\Imo}{\I - 1};
   \draw[thin, color = black!30!green, <-] (u-\I) to (u-\Imo);
   \draw[thin, color = black!30!green, ->] (v-\I) to (v-\Imo);
}

\draw[thin, color = black!30!green, ->] (s) to (u-0);
\draw[thin, color = black!30!green, <-] (t) to (v-0);

\node[vertex] (sm) at (6,5.4) {};
\node[vertex] (um) at (6.6,5.4) {};
\node[vertex] (vm) at (7.2,5.4) {};
\node[vertex] (tm) at (7.8,5.4) {};

\draw[thin] (s) to (sm);
\draw[tblocked] (u-0) to (um);
\draw[thin] (v-0) to (vm);
\draw[thin] (t) to (tm);

\draw[thin] (sm) to (um);
\draw[thin] (um) to (vm);
\draw[thin] (vm) to (tm);

\node[scale = 3] (d1) at (4.2,6) {$\ldots$};
\node[scale = 3] at (4.2,9) {$\ldots$};
\node[scale = 3] (d2) at (9.6,6) {$\ldots$};
\node[scale = 3] at (9.6,9) {$\ldots$};
\node[scale=3, rotate=90] at (6.9,4) {$\ldots$};

\node[vertex] (a-0) at (2.4,\B) {};
\foreach \I in {1,...,\l} {
	\pgfmathsetmacro{\J}{\B + (0.6 * \I)};
	\pgfmathtruncatemacro{\Imo}{\I - 1};
	\node[vertex] (a-\I) at (2.4,\J) {};
	\draw[thin] (a-\I) -- (a-\Imo);
}
\node[vertex] (z-0) at (11.4,\B) {};
\foreach \I in {1,...,\l} {
	\pgfmathsetmacro{\J}{\B + 0.6 * \I};
	\pgfmathtruncatemacro{\Imo}{\I - 1};
	\node[vertex] (z-\I) at (11.4,\J) {};
	\draw[thin] (z-\I) -- (z-\Imo);
}

\node[vertex] (t-0) at (2.4,\T) {};
\foreach \I in {1,...,\l} {
	\pgfmathsetmacro{\x}{2.4 + 0.6 * \I};
	\pgfmathtruncatemacro{\Imo}{\I - 1};
	\node[vertex] (t-\I) at (\x , \T) {};
	\draw[thin] (t-\I) -- (t-\Imo);
}

\node[vertex] (b-0) at (2.4,\B) {};
\foreach \I in {1,...,\l} {
	\pgfmathsetmacro{\x}{2.4 + 0.6 * \I};
	\pgfmathtruncatemacro{\Imo}{\I - 1};
	\node[vertex] (b-\I) at (\x , \B) {};
	\draw[thin] (b-\I) -- (b-\Imo);
}

\draw[thin, color = black!30!green, <-] (t-7) -- (u-7);
\draw[thin, color = black!30!green, ->] (t-8) -- (v-7);
\draw[thin, color = black!30!green, ->] (t-7) -- (t-8);

\end{tikzpicture}}
	\caption{Road map on which the \textsc{Greedy} strategy achieves ratio $\Omega(k)$.}
	\label{subfig:greedy}
	\end{subfigure}
	\qquad
	\begin{subfigure}[b]{0.40\columnwidth}
	\centering
	\scalebox{0.6}{\begin{tikzpicture}

\def\n{9} 
\def\m{9} 
\pgfmathtruncatemacro{\nmo}{\n-1}

\foreach \I in {1,...,\n} {
    \foreach \J in {1,...,\m} {
        \node[vertex] (v-\I-\J) at (\I,\J) {};
    }
}

\draw[line width = 7pt, color=white!40!orange] (v-1-\m) -- (1,3) -- (\n,3) -- (v-\n-1);

\foreach \I in {2,...,\nmo} {
    \foreach \J in {1,...,\m} {
        \pgfmathtruncatemacro{\Imo}{\I-1} 
        \draw[free] (v-\I-\J) to (v-\Imo-\J);
    }
}

\foreach \I in {2,...,\nmo} {
    \foreach \J in {2,...,\m} {
        \pgfmathtruncatemacro{\Jmo}{\J-1} 
        \draw[tblocked] (v-\I-\J) to (v-\I-\Jmo);
    }
}

\foreach \J in {1,...,2} {
    \draw[tblocked] (v-\n-\J) to (v-\nmo-\J);
}

\draw[free] (v-\n-3) to (v-\nmo-3);

\foreach \J in {4,...,\n} {
    \draw[tblocked] (v-\n-\J) to (v-\nmo-\J);
}

\foreach \J in {2,...,\m} {
    \pgfmathtruncatemacro{\Jmo}{\J-1} 
    \draw[free] (v-1-\J) to (v-1-\Jmo);
    \draw[free] (v-\n-\J) to (v-\n-\Jmo);
}

\node[terminal] (s) at (1,\n) {};
\draw (s) node[above left,scale=2,color = orange] {$s$};
\node[terminal] (t) at (\n,1) {};
\draw (t) node[above right,scale=2,color = orange] {$t$};

\end{tikzpicture}}
	\caption{Road map on which any strategy achieves at least ratio $\Omega(\sqrt{k})$.}
	\label{subfig:lowerbound}
	\end{subfigure}
	\caption{Lower bounds instances of Section~\ref{sec:insight}}
	\label{fig:lowerbound}
\end{figure}

An example of such a road map is given in Figure~\ref{fig:lowerbound}. The blocked edges are drawn in red, and the open edges appear with a thick black line. The optimal offline path is highlighted in orange. A first remark is that the number of blocked edges is exactly $N(N-1) + N = N^2$, so these road maps are \kctp\ instances for $k = N^2$. A second comment is that, thanks to the single open horizontal edge on the rightmost column, there is an $(s,t)$-path in graph $G\setminus E_*$ and the shortest one has length $2N$.

Consider any deterministic strategy $A$: it must, at some moment, reach the rightmost vertical line, {\em i.e.} a vertex $v$ such that $\ve{x}(v) = N$. To do so, it makes the traveller traverse a horizontal edge adjacent to this rightmost vertical line. Consider the configuration where the single open edge among these horizontal edges corresponds to the last attempt made by strategy $A$. Said differently, the traveller discovers $N$ horizontal blockages before reaching the rightmost vertical line. Because of the vertical blockages forcing the traveller to traverse the length of the grid in order to reach the rightmost vertical line, they have no choice but retrace their steps towards the leftmost vertical line.

For these reasons, the distance traversed by the traveller guided by strategy $A$ is at least $N$ (distance needed to attain the $\ve{y}$-coordinate equal to $0$) plus $2N(N-1) + N$ for all the horizontal paths traversed. Hence, the total distance traversed is at least $2N^2$. Consequently, the competitive ratio of strategy $A$ on these road maps is at least $\frac{2N^2}{2N} = N =\sqrt{k}$.
\end{proof}


Although we assume that $s$ and $t$ lie at opposite corners of the square, the proof extends to arbitrary positions by adding at most $2N$ blockages to restrict the accessible area. This does not affect the asymptotic analysis, since the distance traversed is $\Omega(N^2)$. We now prove that the resulting $\sqrt{k}$ competitive ratio is optimal for unit-weighted solid grid graphs.

\section{Competitive analysis when the grid size is small} \label{sec:strategy}

Let us consider now, as announced, $G$ as a solid grid graph, with $s,t \in V(G)$, and $k$ a positive integer. This defines the visible part of a $k$-CTP grid instance. Let $N = d_{\emptyset}(s,t)$. 
We also fix $E_*$ as a set of blocked edges to be discovered by the traveller, $\vert E_* \vert \le k$. 

Our first objective in this section is to introduce the strategy \flw. 
Then, we prove that its competitive ratio on solid grid graphs is $O(\sqrt{k})$ when the total size of the grid $\vert E \vert$ is at most $\alpha \cdot N^2$, where $\alpha$ is any constant positive value, {\em i.e.} $\alpha = O(1)$. The case where the grid size is arbitrarily larger than $N$ will be investigated later (Section~\ref{sec-double}).


\subsection{Description of the \flw\ strategy}


We introduce some vertex coloring, which depends on the set $E_*$. 

\begin{definition}
A vertex is called \emph{red} if it is an endpoint of an edge in $E_*$, and \emph{blue} if it is not red but adjacent to a red vertex. We fix $R$ and $B$ the sets of red and blue vertices, respectively. We say a \emph{colored} vertex is either blue or red.

The \emph{red-blue graph} is the subgraph of $G\setminus E_*$ induced by
$R\cup B$. For a colored vertex $v$, we denote by $\rb(v)$ the connected
component of $v$ in the red-blue graph.
\label{def:red-blue}
\end{definition}

Naturally, the red-blue graph is not necessarily connected, as blocked edges can be located anywhere in graph $G$. Hence, it admits connected components which will be treated quite independently by our algorithm \flw .

\begin{observation}
\label{obs:rb-disjoint}
Any two distinct red-blue components are vertex-disjoint. In particular,
if $u$ and $v$ are colored vertices, then either
$\rb(u)=\rb(v)$ or $\rb(u)\cap \rb(v)=\emptyset$.
\end{observation}
%

%
%
%


We now describe our strategy \flw . The strategy starts by fixing a shortest $(s,t)$-path $P=(v_0=s,v_1,\ldots,v_q=t)$ in $G$. The traveller traverses $P$, whenever edges are open. If it encounters a blockage
$v_jv_{j+1}\in E_*$ that has not yet been discovered, it initiates an
exhaustive exploration of $\rbm(v_j)$ following a Depth-First Search (DFS). We will show in Theorem~\ref{thm:rb-routing}, presented in the next subsection, that this exploration
produces an open (colored) path from $v_j$ to a later vertex $v_\ell$ of $P$, $\ell > j$, such
that either $v_\ell=t$ or $v_\ell v_{\ell+1}\notin E_*$. The strategy follows the open colored path inside $\rbm(v_j)$ to such a $v_\ell$ which is as close as possible from $t$. Then, the traveller restarts the walk towards $t$ on $P$, until finding a new red-blue component.

\begin{algorithm}[ht]
\caption{\textsc{RB-DFS}}\label{algo:dfs}
\KwIn{Grid $G$, target $t$, blockages already discovered $E_*' \subsetneq E$, position $v_i$, path $P$}
$T \leftarrow$ tree containing a single vertex $v_i$; $\textsf{parent} \gets v_i$\tcp*{DFS tree}
Let $\textsf{stack}$ be the DFS stack; Push$(\textsf{stack},v_i)$\tcp*{DFS stack}
Label all vertices as undiscovered\;
\While{\textsf{stack} is nonempty }{
	$x \gets$ Pop$(\textsf{stack})$\;
	\If{$x$ is not labeled as discovered}{
		add vertex $x$ and edge $(x,\textsf{parent})$ to $T$; $\textsf{parent} \gets x$\;
		the traveller goes to $x$ (if not already there) via $T$; Label $x$ as discovered\;
		Decide the color of $x$ (visit the four neighbors if necessary) \label{line-decide}\;
        \If{$x$ is colored}{
            \For{all neighbors $y$ of $x$ in $G \setminus E_*'$ not already present in \textsf{stack}}{
            Push$(\textsf{stack},y)$}
        }
	}
}
Update $E_*'$ with all the blockages discovered during the exploration\;
$\ell \gets i$ \tcp*{Determining the vertex of $\rb(v_i) \cap P$ which is the closest to $t$}
\For{all vertices $v_j \in T \cap P$}{
	\lIf{$d_{E_*'}(v_j,t) < d_{E_*'}(v_\ell,t)$ and $v_jv_{j+1} \notin E_*$}{$\ell = j$}
}
\lIf{$\ell = i$}{\Return{``Error: no $(s,t)$-path in $G\setminus E_*$''} \label{line-error}}
The traveller goes to $v_{\ell}$ following tree $T$; \Return{$v_{\ell}$} \label{line:return-dfs}\;
\end{algorithm}

Before giving the full pseudocode of \flw , we define the subroutine called \textsc{RB-DFS} (Algorithm~\ref{algo:dfs}). This procedure guides the traveller in exploring a red-blue component, following a DFS tree. It takes as input the graph $G$ with a set of already discovered blocked edges $E_*'$, a current position $v_i$ and the $(s,t)$-path $P$. It makes the traveller visit all vertices of $\rbm(v_i)$ and returns the closest-to-$t$ vertex $v_{\ell}$, $\ell > i$, belonging both to $P$ and $\rbm(v_i)$. \textsc{RB-DFS} cannot be seen as a classical DFS where all the information is known from the beginning of the traversal. Indeed, the traveller discovers little by little adjacent blockages and identifies whether vertices are red, blue, or uncolored. Hence, each vertex status and the content of the red-blue component is incrementally updated while the exploration goes on.

Concretely, the \flw\ strategy consists in repeating \textsc{RB-DFS} each time a blockage is discovered on the shortest $(s,t)$-path $P$. Indeed, each time \textsc{RB-DFS} is launched on a red-blue component $\rbm(v_i)$, it necessarily meets another vertex $v_{\ell} \in P$ such that the section between $v_{\ell}$ and $t$ on path $P$ is apparently open. And therefore, $v_{\ell}$ is necessarily closer to $t$ than $v_i$. This is a result we are going to state and prove later (Theorem~\ref{thm:rb-routing}). Note, as a technical detail, that for now, the pseudocode of \textsc{RB-DFS} is consistent, since we define $v_{\ell}$ as the closest-to-$t$ vertex of $\rbm(v_i) \cap P$: such a vertex necessarily exists as $v_i$ is candidate.

\begin{algorithm}
\caption{\flw}\label{algo:flw}
\KwIn{Grid $G$, source $s$, target $t$}
$P \gets$ any shortest $(s,t)$-path in $G$ \tcp*{Definition of $P = (s=v_0,v_1,\ldots,v_q=t)$}
$E_*' \gets \emptyset$; $i \gets 0$ \tcp*{Blocked edges discovered and position iterator}
\While{$i < q$}{
	\tcp{Run until reaching $t$}
	\While{$(v_i,v_{i+1})$ is open and $i < q$}{
		\For{all blockages $(v_i,u) \in E_*$ adjacent to $v_i$}{
		$E_*' \gets E_*' \cup \{(v_i,u)\}$ \tcp*{Adding the discovered blockage to $E_*'$}
		Color $v_i$ and $u$ in red, their neighbors in blue if they are not colored yet\;
		}
		Traverse $(v_i,v_{i+1})$; $i \gets i+1$\;
	}
	\lIf{$i = q$}{\Return{}}
	$v_{\ell} \gets \textsc{RB-DFS}(G,t,E_*',v_i,P)$ \tcp*{Explore $\rb(v_i)$ when $(v_i,v_{i+1})$ is blocked}
	$i \gets \ell$ \tcp*{Theorem~\ref{thm:rb-routing} shows that $\ell > i$}
}
\end{algorithm}

The pseudocode of \flw\ is provided in Algorithm~\ref{algo:flw}. Observe that, as long as path $P$ is open, the traveller advances on it and maintains both set $E_*'$ and color statuses.

This strategy \flw\ is illustrated in Figure~\ref{fig:flw}, which describes its different steps on a solid grid example. The position of the traveller is indicated by a green line surrounding the vertex. The traveller starts on vertex $s$ in Figure~\ref{subfig:flw1}, and the pre-computed path $P$ is highlighted in orange. Then, the traveller traverses path $P$ until reaching a blockage $(u,v) \in E_*$ on $P$ (Figure~\ref{subfig:flw2}). At this point, subroutine \textsc{RB-DFS} is launched: Figure~\ref{subfig:flw3} illustrates the beginning of the exploration of the red-blue component $\rbm(u)$ and extra blockages are discovered. Green arcs and numbers indicate in which order the vertices of $\rbm(u)$ are visited by the \textsc{RB-DFS}. Observe that, due to line~\ref{line-decide} of Algorithm~\ref{algo:dfs}, there are additional edges actually traversed by the traveller. Indeed, for each vertex visited, deciding its color amounts to exploring all its neighbors to check whether they are adjacent (or not) to a blockage.
The \textsc{RB-DFS} exploration goes on and is completed in Figure~\ref{subfig:flw4}, with all red and blue vertices $\rbm(u)$ explored. Then, the traveller moves in the DFS tree to join the vertex of $\rbm(u)$ which is the closest to $t$ (Figure~\ref{subfig:flw5}). The traveller continues their journey towards $t$ and is luckily not blocked in this example (Figure~\ref{subfig:flw6}).

\begin{figure}[ht]
	\centering
	\begin{subfigure}[b]{0.32\columnwidth}
	\scalebox{0.5}{\begin{tikzpicture}

\def\n{9} 
\def\m{9} 
\pgfmathtruncatemacro{\nmo}{\n-1}
\pgfmathtruncatemacro{\mmo}{\m-2}

\foreach \I in {4,...,7} {
    \foreach \J in {1,...,\m} {
        \node[vertex] (v-\I-\J) at (\I,\J) {};
    }
}
\foreach \J in {2,...,\m} {
    \node[vertex] (v-3-\J) at (3,\J) {};
    \node[vertex] (v-8-\J) at (8,\J) {};
}

\foreach \J in {4,...,\m} {
    \node[vertex] (v-1-\J) at (1,\J) {};
    \node[vertex] (v-2-\J) at (2,\J) {};
}

\foreach \J in {2,...,\mmo} {
    \node[vertex] (v-9-\J) at (9,\J) {};
}

\draw[line width = 5pt, color=white!40!orange] (2,8) -- (3,8) -- (4,8) -- (5,8) -- (6,8) -- (6,7) -- (6,6) -- (6,5) -- (6,4) -- (6,3)-- (6,2);

\foreach \I in {2,...,8} {
    \foreach \J in {4,...,\m} {
        \pgfmathtruncatemacro{\Imo}{\I-1} 
        \draw[thin] (v-\I-\J) to (v-\Imo-\J);
    }
}

\foreach \J in {4,...,7} {
        \draw[thin] (v-9-\J) to (v-8-\J);
}

\foreach \I in {4,...,9} {
        \pgfmathtruncatemacro{\Imo}{\I-1} 
        \draw[thin] (v-\I-3) to (v-\Imo-3);
        \draw[thin] (v-\I-2) to (v-\Imo-2);
}

\foreach \I in {5,...,7} {
        \pgfmathtruncatemacro{\Imo}{\I-1} 
        \draw[thin] (v-\I-1) to (v-\Imo-1);
}

\foreach \I in {4,...,7} {
    \foreach \J in {2,...,\m} {
        \pgfmathtruncatemacro{\Jmo}{\J-1} 
        \draw[thin] (v-\I-\J) to (v-\I-\Jmo);
    }
}

\foreach \J in {5,...,\m} {
    \pgfmathtruncatemacro{\Jmo}{\J-1} 
    \draw[thin] (v-1-\J) to (v-1-\Jmo);
    \draw[thin] (v-2-\J) to (v-2-\Jmo);
}

\foreach \J in {3,...,\m} {
    \pgfmathtruncatemacro{\Jmo}{\J-1} 
    \draw[thin] (v-3-\J) to (v-3-\Jmo);
    \draw[thin] (v-8-\J) to (v-8-\Jmo);
}

\foreach \J in {3,...,7} {
    \pgfmathtruncatemacro{\Jmo}{\J-1} 
    \draw[thin] (v-9-\J) to (v-9-\Jmo);
}

\node[terminal, draw = black] at (2,8) {};2
\node[terminal, scale = 1.2, draw = lime, line width = 4pt] (s) at (2,8) {};
\draw (s) node[above left,scale=2,color = orange] {$s$};

\node[terminal] (t) at (6,2) {};
\draw (t) node[below right,scale=2,color = orange] {$t$};

%
%
%

\end{tikzpicture}}
	\caption{Solid grid instance for \kctp\ with a pre-computed path $P$}
	\label{subfig:flw1}
	\end{subfigure}
	\begin{subfigure}[b]{0.32\columnwidth}
	\scalebox{0.5}{\begin{tikzpicture}

\def\n{9} 
\def\m{9} 
\pgfmathtruncatemacro{\nmo}{\n-1}
\pgfmathtruncatemacro{\mmo}{\m-2}

\foreach \I in {4,...,7} {
    \foreach \J in {1,...,\m} {
        \node[vertex] (v-\I-\J) at (\I,\J) {};
    }
}
\foreach \J in {2,...,\m} {
    \node[vertex] (v-3-\J) at (3,\J) {};
    \node[vertex] (v-8-\J) at (8,\J) {};
}

\foreach \J in {4,...,\m} {
    \node[vertex] (v-1-\J) at (1,\J) {};
    \node[vertex] (v-2-\J) at (2,\J) {};
}

\foreach \J in {2,...,\mmo} {
    \node[vertex] (v-9-\J) at (9,\J) {};
}

\draw[line width = 5pt, color=white!40!orange] (v-6-6) -- (6,5) -- (6,4) -- (6,3) -- (6,2);

\foreach \I in {2,...,8} {
    \foreach \J in {4,...,\m} {
        \pgfmathtruncatemacro{\Imo}{\I-1} 
        \draw[thin] (v-\I-\J) to (v-\Imo-\J);
    }
}

\foreach \J in {4,...,7} {
        \draw[thin] (v-9-\J) to (v-8-\J);
}

\foreach \I in {4,...,9} {
        \pgfmathtruncatemacro{\Imo}{\I-1} 
        \draw[thin] (v-\I-3) to (v-\Imo-3);
        \draw[thin] (v-\I-2) to (v-\Imo-2);
}

\foreach \I in {5,...,7} {
        \pgfmathtruncatemacro{\Imo}{\I-1} 
        \draw[thin] (v-\I-1) to (v-\Imo-1);
}

\foreach \I in {4,...,7} {
    \foreach \J in {2,...,\m} {
        \pgfmathtruncatemacro{\Jmo}{\J-1} 
        \draw[thin] (v-\I-\J) to (v-\I-\Jmo);
    }
}

\foreach \J in {5,...,\m} {
    \pgfmathtruncatemacro{\Jmo}{\J-1} 
    \draw[thin] (v-1-\J) to (v-1-\Jmo);
    \draw[thin] (v-2-\J) to (v-2-\Jmo);
}

\foreach \J in {3,...,\m} {
    \pgfmathtruncatemacro{\Jmo}{\J-1} 
    \draw[thin] (v-3-\J) to (v-3-\Jmo);
    \draw[thin] (v-8-\J) to (v-8-\Jmo);
}

\foreach \J in {3,...,7} {
    \pgfmathtruncatemacro{\Jmo}{\J-1} 
    \draw[thin] (v-9-\J) to (v-9-\Jmo);
}

\draw[->, color = black!30!green, line width = 2pt] (v-2-8) -- (v-3-8);
\draw[->, color = black!30!green, line width = 2pt] (v-3-8) -- (v-4-8);
\draw[->, color = black!30!green, line width = 2pt] (v-4-8) -- (v-5-8);
\draw[->, color = black!30!green, line width = 2pt] (v-5-8) -- (v-6-8);
\draw[->, color = black!30!green, line width = 2pt] (v-6-8) -- (v-6-7);
\node[scale=0.8,draw,circle,fill=white,draw=lime] at (6,7) {}; //six seven
\node[scale=1,draw,circle,fill=white,draw=lime, line width = 4pt] (v-6-7) at (6,7) {};
\draw[blocked] (v-6-7) -- (v-6-6);

\node[terminal] (s) at (2,8) {};
\draw (s) node[above left,scale=2,color = orange] {$s$};

\node[terminal] (t) at (6,2) {};
\draw (t) node[below right,scale=2,color = orange] {$t$};

%
%
%

\end{tikzpicture}}
	\caption{The traveller (in green) moves on $P$ but is blocked}
	\label{subfig:flw2}
	\end{subfigure}
	\begin{subfigure}[b]{0.32\columnwidth}
	\scalebox{0.5}{\begin{tikzpicture}

\def\n{9} 
\def\m{9} 
\pgfmathtruncatemacro{\nmo}{\n-1}
\pgfmathtruncatemacro{\mmo}{\m-2}

\foreach \I in {4,...,7} {
    \foreach \J in {1,...,\m} {
        \node[vertex] (v-\I-\J) at (\I,\J) {};
    }
}
\foreach \J in {2,...,\m} {
    \node[vertex] (v-3-\J) at (3,\J) {};
    \node[vertex] (v-8-\J) at (8,\J) {};
}

\foreach \J in {4,...,\m} {
    \node[vertex] (v-1-\J) at (1,\J) {};
    \node[vertex] (v-2-\J) at (2,\J) {};
}

\foreach \J in {2,...,\mmo} {
    \node[vertex] (v-9-\J) at (9,\J) {};
}

\draw[line width = 5pt, color=white!40!orange] (v-6-6) -- (6,5) -- (6,4) -- (6,3) -- (6,2);

\foreach \I in {2,...,8} {
    \foreach \J in {4,...,\m} {
        \pgfmathtruncatemacro{\Imo}{\I-1} 
        \draw[thin] (v-\I-\J) to (v-\Imo-\J);
    }
}

\foreach \J in {4,...,7} {
        \draw[thin] (v-9-\J) to (v-8-\J);
}

\foreach \I in {4,...,9} {
        \pgfmathtruncatemacro{\Imo}{\I-1} 
        \draw[thin] (v-\I-3) to (v-\Imo-3);
        \draw[thin] (v-\I-2) to (v-\Imo-2);
}

\foreach \I in {5,...,7} {
        \pgfmathtruncatemacro{\Imo}{\I-1} 
        \draw[thin] (v-\I-1) to (v-\Imo-1);
}

\foreach \I in {4,...,7} {
    \foreach \J in {2,...,\m} {
        \pgfmathtruncatemacro{\Jmo}{\J-1} 
        \draw[thin] (v-\I-\J) to (v-\I-\Jmo);
    }
}

\foreach \J in {5,...,\m} {
    \pgfmathtruncatemacro{\Jmo}{\J-1} 
    \draw[thin] (v-1-\J) to (v-1-\Jmo);
    \draw[thin] (v-2-\J) to (v-2-\Jmo);
}

\foreach \J in {3,...,\m} {
    \pgfmathtruncatemacro{\Jmo}{\J-1} 
    \draw[thin] (v-3-\J) to (v-3-\Jmo);
    \draw[thin] (v-8-\J) to (v-8-\Jmo);
}

\foreach \J in {3,...,7} {
    \pgfmathtruncatemacro{\Jmo}{\J-1} 
    \draw[thin] (v-9-\J) to (v-9-\Jmo);
}

\node[scale=0.8,draw,circle,fill=blue,draw=lime] at (6,8) {};
\node[scale=1,draw,circle,fill=blue,draw=lime, line width  = 3pt] (v-6-8) at (6,8) {};
\node[scale=0.8,draw,circle,fill=blue,draw=black] (v-7-8) at (7,8) {};
\node[scale=0.8,draw,circle,fill=blue,draw=black] (v-8-8) at (8,8) {};
\node[scale=0.8,draw,circle,fill=blue,draw=black] (v-5-7) at (5,7) {};


\draw[->, color = black!30!green, line width = 2pt] (v-6-7) to node[above, color = black!30!green, scale = 0.9] {1} (v-7-7);
\draw[->, color = black!30!green, line width = 2pt] (v-7-7) to node[above, color = black!30!green, scale = 0.9] {2} (v-8-7);
\draw[<->, color = black!30!green, line width = 2pt] (v-8-7) to node[above, color = black!30!green, scale = 0.9] {3} (v-9-7);
\draw[->, color = black!30!green, line width = 2pt] (v-8-7) to node[right, color = black!30!green, scale = 0.9] {4} (v-8-8);
\draw[->, color = black!30!green, line width = 2pt] (v-8-8) to node[above, color = black!30!green, scale = 0.9] {5} (v-7-8);
\draw[->, color = black!30!green, line width = 2pt] (v-7-8) to node[above, color = black!30!green, scale = 0.9] {6} (v-6-8);

\foreach \I in {6,...,\n} {
  \node[scale=0.8,draw,circle,fill=red,draw=black] (v-\I-7) at (\I,7) {};
}


\draw[blocked] (v-6-7) -- (v-6-6);
\draw[blocked] (v-7-7) -- (v-7-6);
\draw[blocked] (v-8-7) -- (v-8-6);
\draw[blocked] (v-9-7) -- (v-9-6);

\node[terminal] (s) at (2,8) {};
\draw (s) node[above left,scale=2,color = orange] {$s$};

\node[terminal] (t) at (6,2) {};
\draw (t) node[below right,scale=2,color = orange] {$t$};

\end{tikzpicture}}
	\caption{The traveller starts exploring the red-blue component}
	\label{subfig:flw3}
	\end{subfigure}
	
	\begin{subfigure}[b]{0.32\columnwidth}
	\scalebox{0.5}{\begin{tikzpicture}

\def\n{9} 
\def\m{9} 
\pgfmathtruncatemacro{\nmo}{\n-1}
\pgfmathtruncatemacro{\mmo}{\m-2}

\foreach \I in {4,...,7} {
    \foreach \J in {1,...,\m} {
        \node[vertex] (v-\I-\J) at (\I,\J) {};
    }
}
\foreach \J in {2,...,\m} {
    \node[vertex] (v-3-\J) at (3,\J) {};
    \node[vertex] (v-8-\J) at (8,\J) {};
}

\foreach \J in {4,...,\m} {
    \node[vertex] (v-1-\J) at (1,\J) {};
    \node[vertex] (v-2-\J) at (2,\J) {};
}

\foreach \J in {2,...,\mmo} {
    \node[vertex] (v-9-\J) at (9,\J) {};
}

\draw[line width = 5pt, color=white!40!orange] (v-6-5) -- (6,4) -- (6,3) -- (6,2);

\foreach \I in {2,...,8} {
    \foreach \J in {4,...,\m} {
        \pgfmathtruncatemacro{\Imo}{\I-1} 
        \draw[thin] (v-\I-\J) to (v-\Imo-\J);
    }
}

\foreach \J in {4,...,7} {
        \draw[thin] (v-9-\J) to (v-8-\J);
}

\foreach \I in {4,...,9} {
        \pgfmathtruncatemacro{\Imo}{\I-1} 
        \draw[thin] (v-\I-3) to (v-\Imo-3);
        \draw[thin] (v-\I-2) to (v-\Imo-2);
}

\foreach \I in {5,...,7} {
        \pgfmathtruncatemacro{\Imo}{\I-1} 
        \draw[thin] (v-\I-1) to (v-\Imo-1);
}

\foreach \I in {4,...,7} {
    \foreach \J in {2,...,\m} {
        \pgfmathtruncatemacro{\Jmo}{\J-1} 
        \draw[thin] (v-\I-\J) to (v-\I-\Jmo);
    }
}

\foreach \J in {5,...,\m} {
    \pgfmathtruncatemacro{\Jmo}{\J-1} 
    \draw[thin] (v-1-\J) to (v-1-\Jmo);
    \draw[thin] (v-2-\J) to (v-2-\Jmo);
}

\foreach \J in {3,...,\m} {
    \pgfmathtruncatemacro{\Jmo}{\J-1} 
    \draw[thin] (v-3-\J) to (v-3-\Jmo);
    \draw[thin] (v-8-\J) to (v-8-\Jmo);
}

\foreach \J in {3,...,7} {
    \pgfmathtruncatemacro{\Jmo}{\J-1} 
    \draw[thin] (v-9-\J) to (v-9-\Jmo);
}

\node[scale=0.8,draw,circle,fill=blue,draw=black] (v-6-8) at (6,8) {};
\node[scale=0.8,draw,circle,fill=blue,draw=black] (v-7-8) at (7,8) {};
\node[scale=0.8,draw,circle,fill=blue,draw=black] (v-8-8) at (8,8) {};
\node[scale=0.8,draw,circle,fill=blue,draw=black] (v-3-7) at (3,7) {};
\node[scale=0.8,draw,circle,fill=blue,draw=black] (v-3-6) at (3,6) {};
\node[scale=0.8,draw,circle,fill=blue,draw=black] (v-5-5) at (5,5) {};
\node[scale=0.8,draw,circle,fill=blue,draw=black] at (4,5) {};
\node[scale=1,draw,circle,fill=blue,draw=lime,line width = 3pt] (v-4-5) at (4,5) {};
\node[scale=0.8,draw,circle,fill=blue,draw=black] (v-7-5) at (7,5) {};
\node[scale=0.8,draw,circle,fill=blue,draw=black] (v-8-5) at (8,5) {};
\node[scale=0.8,draw,circle,fill=blue,draw=black] (v-9-4) at (9,4) {};
\node[scale=0.8,draw,circle,fill=blue,draw=black] (v-4-8) at (4,8) {};
\node[scale=0.8,draw,circle,fill=blue,draw=black] (v-5-8) at (5,8) {};
\node[scale=0.8,draw,circle,fill=blue,draw=black] (v-6-4) at (6,4) {};


\draw[->, color = black!30!green, line width = 2pt] (v-6-7) to node[above, color = black!30!green, scale = 0.9] {1} (v-7-7);
\draw[->, color = black!30!green, line width = 2pt] (v-7-7) to node[above, color = black!30!green, scale = 0.9] {2} (v-8-7);
\draw[<->, color = black!30!green, line width = 2pt] (v-8-7) to node[above, color = black!30!green, scale = 0.9] {3} (v-9-7);
\draw[->, color = black!30!green, line width = 2pt] (v-8-7) to node[right, color = black!30!green, scale = 0.9] {4} (v-8-8);
\draw[->, color = black!30!green, line width = 2pt] (v-8-8) to node[above, color = black!30!green, scale = 0.9] {5} (v-7-8);
\draw[->, color = black!30!green, line width = 2pt] (v-7-8) to node[above, color = black!30!green, scale = 0.9] {6} (v-6-8);
\draw[->, color = black!30!green, line width = 2pt] (v-6-8) to node[above, color = black!30!green, scale = 0.9] {7} (v-5-8);
\draw[->, color = black!30!green, line width = 2pt] (v-5-8) to node[right, color = black!30!green, scale = 0.9] {8} (v-5-7);
\draw[->, color = black!30!green, line width = 2pt] (v-5-7) to node[above, color = black!30!green, scale = 0.9] {9} (v-4-7);
\draw[<->, color = black!30!green, line width = 2pt] (v-4-7) to node[right, color = black!30!green, scale = 0.9] {10} (v-4-8);
\draw[->, color = black!30!green, line width = 2pt] (v-4-7) to node[above, color = black!30!green, scale = 0.9] {11} (v-3-7);
\draw[->, color = black!30!green, line width = 2pt] (v-3-7) to node[right, color = black!30!green, scale = 0.9] {12} (v-3-6);
\foreach \I in {3,...,6}{
	\pgfmathtruncatemacro{\Ip}{\I+1}
	\pgfmathtruncatemacro{\L}{\I+10}
	\draw[->, color = black!30!green, line width = 2pt] (v-\I-6) to node[above, color = black!30!green, scale = 0.9] {\L} (v-\Ip-6);
}
\draw[->, color = black!30!green, line width = 2pt] (v-7-6) to node[right, color = black!30!green, scale = 0.9] {17} (v-7-5);
\draw[<->, color = black!30!green, line width = 2pt] (v-7-5) to node[above, color = black!30!green, scale = 0.9] {18} (v-8-5);
\draw[<->, color = black!30!green, line width = 2pt] (v-8-5) to node[above, color = black!30!green, scale = 0.9] {19} (v-9-5);
\draw[<->, color = black!30!green, line width = 2pt] (v-9-5) to node[right, color = black!30!green, scale = 0.9] {20} (v-9-4);
\draw[->, color = black!30!green, line width = 2pt] (v-7-5) to node[above, color = black!30!green, scale = 0.9] {21} (v-6-5);
\draw[<->, color = black!30!green, line width = 2pt] (v-6-5) to node[right, color = black!30!green, scale = 0.9] {22} (v-6-4);
\draw[->, color = black!30!green, line width = 2pt] (v-6-5) to node[above, color = black!30!green, scale = 0.9] {23} (v-5-5);
\draw[->, color = black!30!green, line width = 2pt] (v-5-5) to node[above, color = black!30!green, scale = 0.9] {24} (v-4-5);

\foreach \I in {4,...,\n} {
  \node[scale=0.8,draw,circle,fill=red,draw=black] (v-\I-7) at (\I,7) {};
}
\foreach \I in {4,...,8} {
  \node[scale=0.8,draw,circle,fill=red,draw=black] (v-\I-6) at (\I,6) {}; //six seven
}

\node[scale=0.8,draw,circle,fill=red,draw=black] (v-6-5) at (6,5) {};
\node[scale=0.8,draw,circle,fill=red,draw=black] (v-9-5) at (9,5) {};

\draw[blocked] (v-6-7) -- (v-6-6);
\draw[blocked] (v-7-7) -- (v-7-6);
\draw[blocked] (v-8-7) -- (v-8-6);
\draw[blocked] (v-9-7) -- (v-9-6);
\draw[blocked] (v-9-6) -- (v-9-5);
\draw[blocked] (v-8-6) -- (v-9-6);
\draw[blocked] (v-5-6) -- (v-5-7);
\draw[blocked] (v-4-6) -- (v-4-7);
\draw[blocked] (v-6-6) -- (v-6-5);
\draw[blocked] (v-7-6) -- (v-8-6);

\node[terminal] (s) at (2,8) {};
\draw (s) node[above left,scale=2,color = orange] {$s$};

\node[terminal] (t) at (6,2) {};
\draw (t) node[below right,scale=2,color = orange] {$t$};

\end{tikzpicture}}
	\caption{The run of \textsc{RB-DFS} for this component is terminated.}
	\label{subfig:flw4}
	\end{subfigure}
	\begin{subfigure}[b]{0.32\columnwidth}
	\scalebox{0.5}{\begin{tikzpicture}

\def\n{9} 
\def\m{9} 
\pgfmathtruncatemacro{\nmo}{\n-1}
\pgfmathtruncatemacro{\mmo}{\m-2}

\foreach \I in {4,...,7} {
    \foreach \J in {1,...,\m} {
        \node[vertex] (v-\I-\J) at (\I,\J) {};
    }
}
\foreach \J in {2,...,\m} {
    \node[vertex] (v-3-\J) at (3,\J) {};
    \node[vertex] (v-8-\J) at (8,\J) {};
}

\foreach \J in {4,...,\m} {
    \node[vertex] (v-1-\J) at (1,\J) {};
    \node[vertex] (v-2-\J) at (2,\J) {};
}

\foreach \J in {2,...,\mmo} {
    \node[vertex] (v-9-\J) at (9,\J) {};
}

\draw[line width = 5pt, color=white!40!orange] (v-6-4) -- (6,3) -- (6,2);

\foreach \I in {2,...,8} {
    \foreach \J in {4,...,\m} {
        \pgfmathtruncatemacro{\Imo}{\I-1} 
        \draw[thin] (v-\I-\J) to (v-\Imo-\J);
    }
}

\foreach \J in {4,...,7} {
        \draw[thin] (v-9-\J) to (v-8-\J);
}

\foreach \I in {4,...,9} {
        \pgfmathtruncatemacro{\Imo}{\I-1} 
        \draw[thin] (v-\I-3) to (v-\Imo-3);
        \draw[thin] (v-\I-2) to (v-\Imo-2);
}

\foreach \I in {5,...,7} {
        \pgfmathtruncatemacro{\Imo}{\I-1} 
        \draw[thin] (v-\I-1) to (v-\Imo-1);
}

\foreach \I in {4,...,7} {
    \foreach \J in {2,...,\m} {
        \pgfmathtruncatemacro{\Jmo}{\J-1} 
        \draw[thin] (v-\I-\J) to (v-\I-\Jmo);
    }
}

\foreach \J in {5,...,\m} {
    \pgfmathtruncatemacro{\Jmo}{\J-1} 
    \draw[thin] (v-1-\J) to (v-1-\Jmo);
    \draw[thin] (v-2-\J) to (v-2-\Jmo);
}

\foreach \J in {3,...,\m} {
    \pgfmathtruncatemacro{\Jmo}{\J-1} 
    \draw[thin] (v-3-\J) to (v-3-\Jmo);
    \draw[thin] (v-8-\J) to (v-8-\Jmo);
}

\foreach \J in {3,...,7} {
    \pgfmathtruncatemacro{\Jmo}{\J-1} 
    \draw[thin] (v-9-\J) to (v-9-\Jmo);
}

\node[scale=0.8,draw,circle,fill=blue,draw=black] (v-6-8) at (6,8) {};
\node[scale=0.8,draw,circle,fill=blue,draw=black] (v-7-8) at (7,8) {};
\node[scale=0.8,draw,circle,fill=blue,draw=black] (v-8-8) at (8,8) {};
\node[scale=0.8,draw,circle,fill=blue,draw=black] (v-3-7) at (3,7) {};
\node[scale=0.8,draw,circle,fill=blue,draw=black] (v-3-6) at (3,6) {};
\node[scale=0.8,draw,circle,fill=blue,draw=black] (v-5-5) at (5,5) {};
\node[scale=0.8,draw,circle,fill=blue,draw=black] (v-4-5) at (4,5) {};
\node[scale=0.8,draw,circle,fill=blue,draw=black] (v-7-5) at (7,5) {};
\node[scale=0.8,draw,circle,fill=blue,draw=black] (v-8-5) at (8,5) {};
\node[scale=0.8,draw,circle,fill=blue,draw=black] (v-9-4) at (9,4) {};
\node[scale=0.8,draw,circle,fill=blue,draw=black] (v-4-8) at (4,8) {};
\node[scale=0.8,draw,circle,fill=blue,draw=black] (v-5-8) at (5,8) {};
\node[scale=0.8,draw,circle,fill=blue,draw=lime] at (6,4) {};
\node[scale=1,draw,circle,fill=blue,draw=lime, line width = 3pt] (v-6-4) at (6,4) {};

\draw[->, color = black!30!green, line width = 2pt] (v-4-5) to (v-5-5);
\draw[->, color = black!30!green, line width = 2pt] (v-5-5) to (v-6-5);
\draw[->, color = black!30!green, line width = 2pt] (v-6-5) to (v-6-4);

\foreach \I in {4,...,\n} {
  \node[scale=0.8,draw,circle,fill=red,draw=black] (v-\I-7) at (\I,7) {};
}
\foreach \I in {4,...,8} {
  \node[scale=0.8,draw,circle,fill=red,draw=black] (v-\I-6) at (\I,6) {}; //six seven
}

\node[scale=0.8,draw,circle,fill=red,draw=black] (v-6-5) at (6,5) {};
\node[scale=0.8,draw,circle,fill=red,draw=black] (v-9-5) at (9,5) {};

\draw[blocked] (v-6-7) -- (v-6-6);
\draw[blocked] (v-7-7) -- (v-7-6);
\draw[blocked] (v-8-7) -- (v-8-6);
\draw[blocked] (v-9-7) -- (v-9-6);
\draw[blocked] (v-9-6) -- (v-9-5);
\draw[blocked] (v-8-6) -- (v-9-6);
\draw[blocked] (v-5-6) -- (v-5-7);
\draw[blocked] (v-4-6) -- (v-4-7);
\draw[blocked] (v-6-6) -- (v-6-5);
\draw[blocked] (v-7-6) -- (v-8-6);

\node[terminal] (s) at (2,8) {};
\draw (s) node[above left,scale=2,color = orange] {$s$};

\node[terminal] (t) at (6,2) {};
\draw (t) node[below right,scale=2,color = orange] {$t$};

\end{tikzpicture}}
	\caption{The traveller goes to the closest vertex to $t$.}
	\label{subfig:flw5}
	\end{subfigure}
	\begin{subfigure}[b]{0.32\columnwidth}
	\scalebox{0.5}{\begin{tikzpicture}

\def\n{9} 
\def\m{9} 
\pgfmathtruncatemacro{\nmo}{\n-1}
\pgfmathtruncatemacro{\mmo}{\m-2}

\foreach \I in {4,...,7} {
    \foreach \J in {1,...,\m} {
        \node[vertex] (v-\I-\J) at (\I,\J) {};
    }
}
\foreach \J in {2,...,\m} {
    \node[vertex] (v-3-\J) at (3,\J) {};
    \node[vertex] (v-8-\J) at (8,\J) {};
}

\foreach \J in {4,...,\m} {
    \node[vertex] (v-1-\J) at (1,\J) {};
    \node[vertex] (v-2-\J) at (2,\J) {};
}

\foreach \J in {2,...,\mmo} {
    \node[vertex] (v-9-\J) at (9,\J) {};
}

\foreach \I in {2,...,8} {
    \foreach \J in {4,...,\m} {
        \pgfmathtruncatemacro{\Imo}{\I-1} 
        \draw[thin] (v-\I-\J) to (v-\Imo-\J);
    }
}

\foreach \J in {4,...,7} {
        \draw[thin] (v-9-\J) to (v-8-\J);
}

\foreach \I in {4,...,9} {
        \pgfmathtruncatemacro{\Imo}{\I-1} 
        \draw[thin] (v-\I-3) to (v-\Imo-3);
        \draw[thin] (v-\I-2) to (v-\Imo-2);
}

\foreach \I in {5,...,7} {
        \pgfmathtruncatemacro{\Imo}{\I-1} 
        \draw[thin] (v-\I-1) to (v-\Imo-1);
}

\foreach \I in {4,...,7} {
    \foreach \J in {2,...,\m} {
        \pgfmathtruncatemacro{\Jmo}{\J-1} 
        \draw[thin] (v-\I-\J) to (v-\I-\Jmo);
    }
}

\foreach \J in {5,...,\m} {
    \pgfmathtruncatemacro{\Jmo}{\J-1} 
    \draw[thin] (v-1-\J) to (v-1-\Jmo);
    \draw[thin] (v-2-\J) to (v-2-\Jmo);
}

\foreach \J in {3,...,\m} {
    \pgfmathtruncatemacro{\Jmo}{\J-1} 
    \draw[thin] (v-3-\J) to (v-3-\Jmo);
    \draw[thin] (v-8-\J) to (v-8-\Jmo);
}

\foreach \J in {3,...,7} {
    \pgfmathtruncatemacro{\Jmo}{\J-1} 
    \draw[thin] (v-9-\J) to (v-9-\Jmo);
}

\draw[->, color = black!30!green, line width = 2pt] (v-6-4) to (v-6-3);
\draw[->, color = black!30!green, line width = 2pt] (v-6-3) to (v-6-2);


%

\draw[blocked] (v-6-7) -- (v-6-6);
\draw[blocked] (v-7-7) -- (v-7-6);
\draw[blocked] (v-8-7) -- (v-8-6);
\draw[blocked] (v-9-7) -- (v-9-6);
\draw[blocked] (v-9-6) -- (v-9-5);
\draw[blocked] (v-8-6) -- (v-9-6);
\draw[blocked] (v-5-6) -- (v-5-7);
\draw[blocked] (v-4-6) -- (v-4-7);
\draw[blocked] (v-6-6) -- (v-6-5);
\draw[blocked] (v-7-6) -- (v-8-6);

\node[terminal] (s) at (2,8) {};
\draw (s) node[above left,scale=2,color = orange] {$s$};

\node[terminal, scale = 1.2, draw = lime, line width = 3pt] (t) at (6,2) {};
\draw (t) node[below right,scale=2,color = orange] {$t$};

\end{tikzpicture}}
	\caption{The traveller reaches $t$ without encountering any other blockage.}
	\label{subfig:flw6}
	\end{subfigure}
	\caption{How \flw\ works on an example}
	\label{fig:flw}
\end{figure}
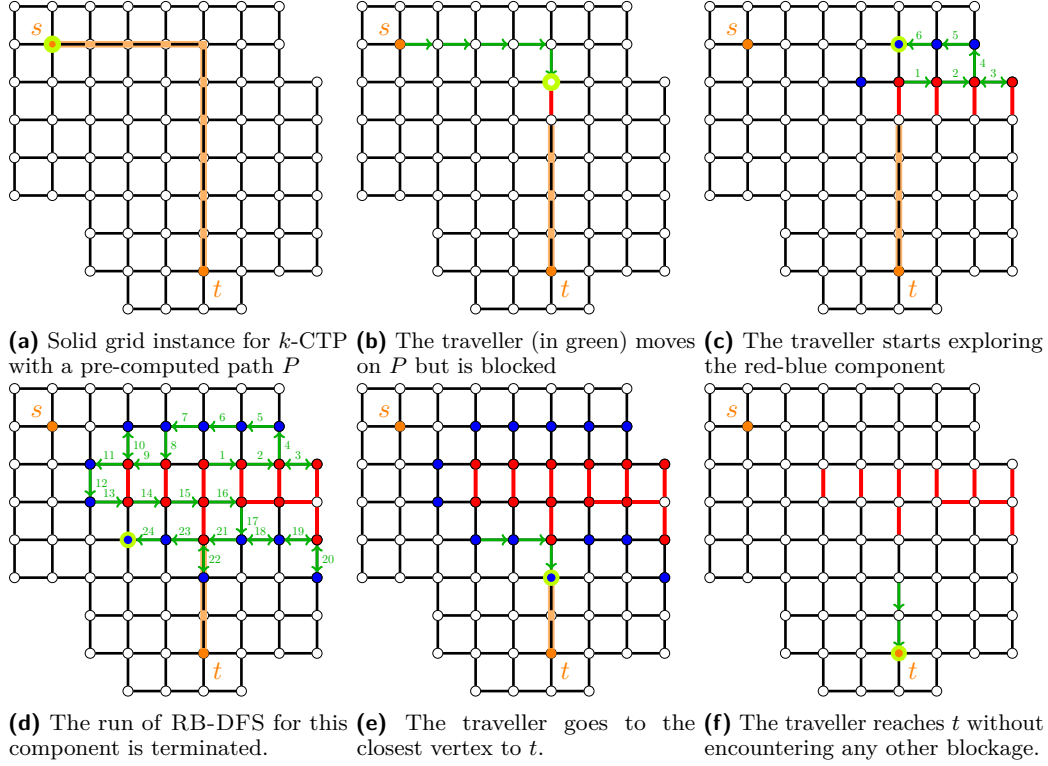

\subsection{Competitive ratio of \flw}

To determine the competitive ratio of \flw\ on solid grid graphs, under the hypothesis that $\vert E \vert \le \alpha N^2$, our reasoning consists in showing that the distance traversed by the traveller is proportional to $k + N$. If such a statement holds, then the competitive ratio of \flw\ depends asymptotically on $\frac{k}{\dopt} \le \frac{k}{N}$. As $k \le \vert E \vert \le \alpha N^2$, we have $N \ge \sqrt{\frac{k}{\alpha}}$, which gives the upper competitive bound $\frac{k}{N} \le \sqrt{\alpha k} = O(\sqrt{k})$. 

Therefore, our first challenge is to assess the distance traversed by the traveller as a function of $k$ and $N$. We begin with a natural property of red-blue components.

\begin{lemma}
\label{lem:red-edge-routing}
Let $uv\in E_*$ be a blockage, then either:
\begin{itemize}
\item there exists a path between $u$ and $v$ included in $\rb(u)$, or
\item there is no open path between $u$ and $v$ in $G \setminus E_*$.
\end{itemize}
\end{lemma}

\begin{proof}
Assume that $uv\in E_*$ and that there exists a path
$Q=(u=u_1,\ldots,u_p=v)$ from $u$ to $v$ in $G\setminus E_*$.
In particular, suppose that $Q$ is chosen to be \emph{area-minimal},
that is, the area of the interior of the cycle formed by adding the edge
$uv$ to $Q$ is minimum. We denote the interior of this cycle by
$\mathcal{C}$ and orient $Q$ from $u$ to $v$ for simplicity.

We prove inductively that $V(Q)\subseteq \rbm(u)$. First,
$u_1=u\in \rbm(u)$. Moreover, $u_2\in \rbm(u)$, since $u$ is red
and $u_2$ is adjacent to $u$, and hence $u_2$ is either red or blue.

Now, assume as induction hypothesis that $u_i\in \rbm(u)$ for some $i<p$. If $u_i$ is red,
then $u_{i+1}$ is blue, since $u_i u_{i+1}\notin E_*$ by definition. By connectivity,
$u_{i+1}\in \rbm(u)$.

Otherwise, suppose that $u_i$ is blue and that $u_{i+1}$ is neither
red nor blue. Consider the vertices $x$ and $y$ such that
$xu_i u_{i+1}y$ is a $4$-cycle and such that the corresponding unit
square is fully contained in $\mathcal{C}$. The vertices $x$ and $y$ exist since $Q$ is a simple cycle in $G$, and $G$ is a solid grid graph and therefore contains no hole. We claim that no edge of
this $4$-cycle is blocked. Indeed, $xu_i\notin E_*$, since otherwise
$u_i$ would be red. Similarly, $u_{i+1}y\notin E_*$, since otherwise
$u_{i+1}$ would be red. Finally, $xy\notin E_*$, since otherwise $u_{i+1}$ would be
blue. All in all, our claim is satisfied.

Therefore, replacing the edge $u_i u_{i+1}$ by the path
$u_i x y u_{i+1}$ yields a path $Q'$ from $u$ to $v$ in
$G\setminus E_*$ whose enclosed area is smaller than the one of
$Q$. This contradicts the area-minimality of $Q$.
Thus, in all cases, $u_{i+1}\in \rbm(u)$. By induction,
$u_i\in \rbm(u)$ for all $i\in\{1,\ldots,p\}$. In particular,
$v = u_p\in \rbm(u)$, and hence there exists a path between $u$ and
$v$ in $\rbm(u)$.
\end{proof}

\flw\ strategy ensures that, after running \textsc{RB-DFS}, the traveller necessarily visits another vertex of $P$ which is closer to $t$. Such a guarantee is formally proven below.

\begin{theorem}[Red-blue routing]
\label{thm:rb-routing}
Let $P=(v_0=s,v_1,\ldots,v_q=t)$ be a shortest $(s,t)$-path in $G$ and assume the traveller follows \flw : they are currently located on $v_j$ and discovers blockage $v_jv_{j+1} \in E_*$.
\textsc{RB-DFS} is then launched to perform an exhaustive
exploration of $\rb(v_j)$ and exactly one of the two following statements holds:
\begin{enumerate}
\item there exists $\ell$, with $j<\ell\le q$, such that $v_\ell \in \rb(v_j)$, and either $v_\ell=t$ or $v_\ell v_{\ell+1}\notin E_*$;
\item there is no $(s,t)$-path in $G \setminus E_*$.
\end{enumerate}
\end{theorem}

\begin{proof}
We start with the proof of the following claim.

\begin{claim}
\label{claim:rb-separator}
Let $\ell \le q$ be the maximum index such that $v_\ell\in \rbm(v_j)$. In other words,
$v_{\ell'}\notin \rbm(v_j)$ for every
$\ell < \ell' \le q$. If
$v_\ell v_{\ell+1}\in E_*$, then there is no $(s,t)$-path in $G \setminus E_*$.
\end{claim}

\begin{claimproof}
Naturally, such an index $\ell$ exists as $j$ is a candidate. We prove by induction on $\ell' > \ell$ that there is no path
between $v_\ell$ and $v_{\ell'}$ in $G \setminus E_*$.

For $\ell'=\ell+1$, we have
$v_\ell v_{\ell+1}\in E_*$. By Lemma~\ref{lem:red-edge-routing},
either there exists a path between $v_\ell$ and $v_{\ell+1}$
contained in $\rbm(v_\ell)$, or there is no open path between
$v_\ell$ and $v_{\ell+1}$ in $G \setminus E_*$. Since $v_\ell\in \rbm(v_j)$, Observation~\ref{obs:rb-disjoint} gives $ \rbm(v_\ell)=\rbm(v_j)$.
The first alternative would imply
$v_{\ell+1}\in \rbm(v_j)$, contradicting our claim hypothesis. So, there is
no $(v_\ell,v_{\ell+1})$-path in $G \setminus E_*$.

Now suppose that the statement holds for every integer between
$\ell$ and $\ell'-1$, where $\ell'>\ell+1$. In particular, for every
$\ell + 1\le r\le \ell' - 1$, there is no open path between
$v_\ell$ and $v_{r}$. We prove that there is no open path between
$v_\ell$ and $v_{\ell'}$.

Suppose, for a contradiction, that a $(v_\ell,v_{\ell'})$-path exists in $G\setminus E_*$. Among
all open paths $Q$ from $v_\ell$ to $v_{\ell'}$, choose one for which
the area enclosed by $Q$ and the subpath
$P[v_\ell,\ldots,v_{\ell'}]$ is minimum, and denote this region by
$\mathcal C$. By the minimality of $\mathcal C$, the paths $Q$ and
$P[v_\ell,\ldots,v_{\ell'}]$ have no common vertex other than
$v_\ell$ and $v_{\ell'}$: otherwise, a proper subpath of $Q$ would
give an open path from $v_\ell$ to some $v_r$, with
$\ell<r<\ell'$, contradicting the induction hypothesis.

We write $
Q=(v_\ell = u_1,\ldots,u_p = v_{\ell'})$.
We prove by induction on $r$ that $
u_r\in \rbm(v_\ell)$ for every $1\le r\le p$. The statement is trivial for $r=1$.
Suppose that $u_r\in \rbm(v_\ell)$ for some $r<p$. Since $Q$ and
$P[v_\ell,\ldots,v_{\ell'}]$ intersect only at their endpoints, the
edge $u_ru_{r+1}$ is incident to a unit square contained in $\mathcal C$. Let $x$ and $y$ be the other two vertices of this
square, so that $xu_ru_{r+1}y$ is its boundary.
At least one edge of this $4$-cycle is blocked. Indeed, if all four edges were open, then replacing $u_ru_{r+1}$ in $Q$ by
the three-edge path $u_rxyu_{r+1}$ would produce an open $(v_\ell,v_{\ell'})$ path enclosing a strictly smaller area than
$\mathcal C$, contradicting the choice of $Q$.

Hence some edge of the $4$-cycle is blocked. Both endpoints of a blocked edge are red, and consequently all four vertices of the
square are red or blue. In particular, $u_{r+1}\in R\cup B$.
Since $u_r\in \rbm(v_\ell)$ and the square contains a red vertex, then $u_{r+1}$ is connected to $u_r$ in the red-blue graph. Thus $u_{r+1}\in \rbm(v_\ell)$.

This completes the induction. In particular, for $r = p$, we have $
v_{\ell'}=u_p\in \rbm(v_\ell)=\rbm(v_j)$.
This contradicts the hypothesis that
$v_{\ell'}\notin \rbm(v_j)$ for every $\ell'>\ell$. Therefore there is
no open path between $v_\ell$ and $v_{\ell'}$.

Taking $\ell'=q$ gives that there is no open path from $v_\ell$ to
$t$. On the other hand, since $v_\ell\in \rbm(v_j)$, there is an open
path from $v_j$ to $v_\ell$ contained in $\rbm(v_j)$. Moreover, there is an $(s,v_j)$-path in $G \setminus E_*$ since the traveller moved from $s$ to $v_j$. Thus, in graph $G \setminus E_*$, vertices $s$ and $v_{\ell}$ belong to the same connected component, but $t$ belongs to another one.
\end{claimproof}

Based on this claim, we fix $\ell$ as the largest index such that $v_\ell\in \rbm(v_j)$.
If $\ell=q$, then
$v_\ell=t$, and the statement of the theorem holds.
Same remark if $\ell<q$ and $v_\ell v_{\ell+1}\notin E_*$. Finally, consider the case $v_\ell v_{\ell+1}\in E_*$: Claim~\ref{claim:rb-separator} implies
that there is no open $(s,t)$-path.
\end{proof}

This latter result is the key property for establishing the competitive ratio of \flw\ on $(G,E_*)$. Concretely, during their walk, the traveller either traverses edges of the path $P$ towards $t$ or explores edges of a red-blue component.

\begin{theorem}\label{thm:starting-strategy}
On a solid grid graph $G$ with blocked edges $E_*$, the \flw\ strategy either makes the traveller reach $t$ if $(G,E_*)$ is feasible or certifies that there is no $(s,t)$-path in $G \setminus E_*$. In both cases, the distance traversed is at most $N + O(k)$.
\end{theorem}

\begin{proof}
Let us first explain why \flw\ actually guides the traveller towards $t$ if $(G,E_*)$ is feasible. The traveller walks on the pre-computed $(s,t)$-path $P = (v_0=s,\ldots,v_q=t)$, until a blockage $v_jv_{j+1}$ is discovered. \textsc{RB-DFS} will explore $\rbm(v_j)$ and, from  Theorem~\ref{thm:rb-routing}, will identify a vertex $v_{\ell} \in P$, $\ell > j$, such that $v_{\ell} \in \rbm(v_j)$. Then, the traveller goes to $v_{\ell}$ (line~\ref{line:return-dfs} of Algorithm~\ref{algo:dfs}). Therefore, \textsc{RB-DFS} moves the traveller to another vertex of $P$ which is closer to $t$. Since $P$ is finite, the traveller will necessarily reach target $t$ after applying the above procedure a finite number of times.
%
%
%
%
%
%

Conversely, if no open $(s,t)$-path exists, then, since the strategy continues along $P$ until reaching $t$ or until a blocked edge is
encountered, some call to Theorem~\ref{thm:rb-routing} must return the second outcome, thereby correctly certifying that there is no $(s,t)$-path in $G \setminus E_*$.

It remains to bound the total length of the walk. The walk alternates
between following the path $P$ and achieving open walks in different red-blue components
$\rbm(v_{j_i})$, with $1\le i\le p$ and $p$ some positive integer. Each such walk is obtained after an
exhaustive exploration of $N_G[\rbm(v_{j_i})]$.
The exploration phase of \textsc{RB-DFS} consists in traversing a tree $T$, and each edge of $T$ is traversed at most twice. Thus, the distance traversed by the traveller for the exploration of $\rbm(v_{j_i})$ is at most $2\vert N_G[\rbm(v_{j_i})]\vert = O(\vert \rbm(v_{j_i}) \vert)$, since the degree of each vertex is at most $4$. To which we can add an extra term $\vert N_G[\rbm(v_{j_i})]\vert = O(\vert \rbm(v_{j_i}) \vert)$ for moving to $v_{\ell_i}$ when the exploration is completed (line~\ref{line:return-dfs} of Algorithm~\ref{algo:dfs}). In the worst case, deciding the color of each vertex explored needs the traveller to traverse each incident edge twice, adding a constant factor of $8$. Altogether, the distance traversed by the traveller due to $\rbm(v_{j_i})$ is $O(\vert\rbm(v_{j_i})\vert)$.

All red-blue components $\rbm(v_{j_i})$ traversed by the traveller are disjoint (Observation~\ref{obs:rb-disjoint}), since at each new launch of \textsc{RB-DFS}, we explore completely the vertex set $\rbm(v_{j_i})$. Thus, the total length of the walk due to all \textsc{RB-DFS} calls is at most $\sum_{i=1}^p O(\vert \rbm(v_{j_i})\vert) = O(\vert R \cup B\vert)$.

Moreover, for each blocked edge, at most $8$ vertices are colored (two red endpoints with $6$ possible neighbors). Thus, $\vert R \cup B\vert \le 8k$.
All in all, the total length traversed by the traveller is at most $\vert V(P) \vert + O(\vert R \cup B\vert) = N + O(k)$.
%
%
\end{proof}

Combining this distance traversed with our initial assumption that $\vert E \vert \le \alpha \cdot N^2$, we conclude that the competitive ratio of the \flw\ strategy is $O(\sqrt{k})$.

\begin{theorem}
On a unit-weighted solid grid graph $G$ with blocked edges $E_*$, where $(G,E_*)$ is feasible and $\vert E \vert \le \alpha \cdot N^2$, the \flw\ strategy has
competitive ratio $O(\sqrt{k})$.
\label{thm-small-grids}
\end{theorem}

\begin{proof}
By Theorem~\ref{thm:starting-strategy}, if $(G,E_*)$ is feasible, then the \flw\ strategy guides the traveller towards $t$ and produces a walk of length at most $N + O(k)$.
By definition, $N = d_{\emptyset}(s,t)$, so $\dopt \ge N$. The competitive ratio of \flw\ on $(G,E_*)$ satisfies:
$$c(G,E_*) = \frac{N + O(k)}{N} = 1 + \frac{O(k)}{N}.$$
Furthermore, $k \le \vert E \vert \le \alpha \cdot N^2$. Reversing the inequality, we have $N \ge \sqrt{\frac{k}{\alpha}}$. As a conclusion,
$$c(G,E_*) = 1 + \frac{O(k)}{N} = 1 + \sqrt{\alpha}~\frac{O(k)}{\sqrt{k}} =O(\sqrt{k}).$$
%
%
%
%
%
%
%
%
%
In brief, the competitive ratio of \flw\ is $O(\sqrt{k})$, seeing $\alpha$ as a constant.
\end{proof}

\section{Extension to all grid sizes} \label{sec-double}

We generalize the result of Theorem~\ref{thm-small-grids}. We still consider unit-weighted solid grids but this time $\vert E \vert$ can be greater than $\alpha \cdot N^2$. We show that successively applying the \flw\ strategy to \emph{balls} of doubling sizes yields a competitive ratio of $O(\sqrt{k})$. We call this strategy \dflw , and its pseudocode is given in Algorithm~\ref{algo:Doublingflw}. Notation $N$ still refers to the initial $(s,t)$-distance $d_{\emptyset}(s,t)$. Remember that $d_{\infty}(u,v)$ is the infinity-norm distance between two vertices $u$ and $v$.

\begin{algorithm}
\caption{\dflw}
\label{algo:Doublingflw}
\KwIn{Solid grid $G$, source $s$, target $t$}
$N \gets d_{\emptyset}(s,t)$; $i_0 \gets \lfloor \log_2(N) \rfloor$; $i\gets 1$\;
$V_1 \gets \{v\in V(G) \mid d_{\infty}(s,v)\leq 2^{i_0+1}\}$ \tcp*{Square ball centered on $s$}
\While{\flw ~fails on $(G[V_i],s,t)$ \label{line-while}}{
the traveller goes back to $s$ using the shortest path among already visited edges \label{line-return-s}\;
$i\gets i+1$\;
$V_i \gets \{v\in V(G) \mid d_{\infty}(s,v)\leq 2^{i_0+i}\}$ \tcp*{Ball size is doubled}
}
\end{algorithm}

Let $i_0 = \lfloor \log_2 N \rfloor$. At each iteration $i \ge 1$, strategy \dflw\ only considers as an input the induced subgraph of $G$ containing vertices at $\ell_\infty$-distance $2^{i_0+i}$ from $s$, {\em i.e.} a value between $2^{i-1}N$ and $2^i N$. In other words, the graph is reduced to the ball centered on $s$ with radius $2^{i_0+i}$. We launch \flw\ on this input: either the traveller reaches $t$ or there is a call of \textsc{RB-DFS} inside \flw\ which returns an error (line~\ref{line-error} of Algorithm~\ref{algo:dfs}), meaning that there is no open $(s,t)$-path inside the ball. In the latter case, we expand the graph to explore, doubling the radius of the ball.

For every $i\geq 1$, let $V_i$ be defined as in Algorithm~\ref{algo:Doublingflw}, namely, as the set of vertices of $G$ at $\ell_{\infty}$-distance at most $2^{i_0+i}$ from $s$. Moreover, we refer to $k_i$ as the number of blocked edges whose endpoints both belong to $V_i$:  $k_i = \vert E_*\cap (V_i \times V_i)\vert$. We also fix $\lambda$ as the total number of iterations achieved by \dflw\ on the input graph: the last grid considered is the $\ell_\infty$-ball of radius $2^{i_0+\lambda}$. We establish a first key property: at each iteration, the number of blockages discovered is proportional to the size of the current ball.

\begin{lemma}[$\star$]\label{lemma:incr-k}
Assume \dflw\ is launched on a solid grid graph $G$, and let $2\leq i <\lambda$. Then
$k_i \geq k_{i-1}+2^{i_0+i-1}$.
\end{lemma}
\begin{proof}[Proof sketch]
The idea behind this proof is relatively simple: as not only $G[V_i]$ but also its subgraph $G[V_{i-1}]$ do not contain any open $(s,t)$-path, a certain number of blockages must be present in $E[V_i] \setminus E[V_{i-1}]$ to separate $s$ from $t$. This latter number is necessarily proportional to the radius of $G[V_i]$, in fact with a factor $\frac{1}{2}$. A full technical version of the proof is available in Appendix~\ref{app:double}, and its major arguments are summarized below.

Suppose, towards a contradiction, that $k_i-k_{i-1}<R$, where $R=2^{i_0+i-1}$. Observe that there are $R$ concentric square boundaries surrounding $V_{i-1}$ and contained in $V_i$, one of which, say $\Gamma$, must contain no blocked edge.
Fix an open $(s,t)$-path $P$ of $G$: it leaves $V_i$ and therefore crosses $\Gamma$ at least twice. We can also fix any $(s,t)$-path $Q$ contained in $V_{i-1}$, which is necessarily blocked. Together, the $(s,t)$-paths $P$ and $Q$ form a cycle $C$. Selecting $P$ and $Q$ such that $C$ minimizes its enclosed area leads to a contradiction: since $G$ is solid, a portion of $\Gamma$ (which is open) provides an open shortcut through the interior of $C$, yielding a cycle with strictly smaller enclosed area.
\end{proof}

Using this lemma, we can now derive the competitive ratio of the strategy. Note that \dflw\ is correct in the sense that it reaches target $t$ on any feasible instance made up of unit-weighted grids. Indeed, Theorem~\ref{thm:starting-strategy} shows that \flw\ either reaches $t$ or certifies that the input is not feasible: as we are considering balls of increasing sizes around $s$, at some point, one ball has to contain an open $(s,t)$-path. Now, it remains to prove our main result: \dflw\ achieves ratio $O(\sqrt{k})$ on all unit-weighted solid grids.

\begin{theorem}
\dflw\ strategy has competitive ratio
$O(\sqrt{k})$ on unit-weighted solid grids.
\label{thm-ratio-all-grids}
\end{theorem}
\begin{proof}
Since there exists an open $(s,t)$-path in $G$, there necessarily exists an index $\lambda$ such that $G[V_{\lambda}]$ contains an open $(s,t)$-path. Let $\lambda$ be the
smallest such index. Hence, \flw\ succeeds on $G[V_{\lambda}]$ and fails on
$G[V_i]$ for every $i<\lambda$.
By minimality of $\lambda$, $\dopt>2^{i_0+\lambda-1}$. 
Indeed, an open $(s,t)$-path of length at most $2^{i_0+\lambda-1}$ cannot leave $V_{\lambda-1}$, since every vertex on such a path is at $\ell_\infty$-distance at most $2^{i_0+\lambda-1}$ from $s$.

Recall by Theorem~\ref{thm:starting-strategy} that, when ran on $G[V_i]$,
strategy \flw\ produces a walk of length $O\!\left(k_i+2^{i_0+i}\right)$ and then either reaches $t$, or states that no open $(s,t)$-path exists in $G[V_i]$. When returning to $s$ (line~\ref{line-return-s} of Algorithm~\ref{algo:Doublingflw}), the traveller traverses at most the distance of the outward journey. Consequently, the total length
$L$ of the walk produced by \dflw\ satisfies:
$
    L
    =O\!\left(
        \sum_{i=1}^{\lambda} \left(k_i+2^{i_0+i}\right)
      \right)
$.
The geometric part of this sum gives $\sum_{i=1}^{\lambda} 2^{i_0+i}=O(2^{i_0+\lambda})=O(\dopt)$, where the last equality follows from
$\dopt>2^{i_0+\lambda-1}$. It therefore remains to prove that: $$\frac{1}{\dopt}\sum_{i=1}^\lambda k_i =O(\sqrt{k}).$$

Since blocked edges outside $G[V_{\lambda}]$ do not affect the execution of \dflw , we can assume w.l.o.g.\ that $k=k_{\lambda}$ for simplicity. We also set $M=2^{i_0+\lambda}$: thus $\dopt>M/2$.
Recall by definition that $G[V_i]$ is contained in a full square grid with $2^{i_0+i+1}+1$ vertices on each side. By grid topology, we have $\vert E(G[V_i]) \vert \leq
    2\left(2^{i_0+i+1}+1\right)2^{i_0+i+1}$.
Since $2^{i_0+i+1}\geq 1$. It follows, using rough upper bounds, that:
\begin{equation}
\label{eq:boundki}
k_i \leq |E(G[V_i])|\leq 2\cdot 4^{i_0+i+2}.
\end{equation}

Note that the case $\lambda=1$ was already considered in Theorem~\ref{thm-small-grids}, because $\vert E[V_1] \vert \le 4N^2$. As another specific case, we consider $\lambda=2$. We have, using Equation~\eqref{eq:boundki}: $$\frac{k_1+k_2}{\dopt} \le 4\frac{k_2}{M} \le 4\sqrt{2} \frac{2^{i_0+4}}{2^{i_0+2}} \sqrt{k_2} = O(\sqrt{k}).$$

From now, we consider the general case $\lambda\geq 3$. We fix a real value $c=2-\frac{\log_2 k}{\log_2 M}$, or said differently $k=M^{2-c}$. We distinguish two scenarii.

\medskip
\noindent
\textbf{Scenario 1: $c\leq 0$.}
Here, $k\geq M^2$. Using $\dopt>M/2$ and $M=2^{i_0+\lambda}$, we obtain
\begin{align*}
    \frac{1}{\dopt}\sum_{i=1}^{\lambda} k_i
    &\leq
    \frac{4^{i_0+3}}{M}\sum_{i=1}^{\lambda} 4^{i} \quad \text{using Eq.(\ref{eq:boundki})}\\
    &\leq\frac{4^{i_0+3}}{M} \frac{4^{\lambda+1}}{3} \leq \frac{4^{4}M^2}{3M} \le \frac{4^4}{3}\sqrt{k}.
\end{align*}

\medskip
\noindent
\textbf{Scenario 2: $c>0$.}
As a first observation, let us show that $c$ is bounded above by $1$, up to an irrelevant additive adjustment. Indeed, by Lemma~\ref{lemma:incr-k},
for every $2\le i< \lambda$, $k_i\geq k_{i-1}+2^{i_0+i-1}$. Iterating this inequality gives $$k = k_r \ge \sum_{i=2}^{\lambda-1}2^{i_0+i-1} = 2^{i_0} (2^{\lambda-2}-1)\geq 2^{i_0+\lambda-3} = \frac{M}{2^3}.$$
where the last inequality holds since $\lambda\geq 3$. It follows that $c=2-\frac{\log_2(k)}{\log_2(M)}\leq 1+\frac{3}{\log_2(M)}$.
Note that value $\frac{3}{\log_2(M)}$ can be assumed to be smaller than $\frac{1}{2}$, as otherwise $M$ is necessarily upper-bounded by a constant value and so is $k$. Thus, we can consider $c \le \frac{3}{2}$.

We now define $\lambda_0=\max\left\{\lambda-\left\lceil c\log_2(k)\right\rceil, 0\right\}$ and split the sum into:
$$
    \frac{1}{\dopt}\sum_{i=1}^\lambda k_i
    =
    \frac{1}{\dopt}\sum_{i=1}^{\lambda_0}k_i
    +
    \frac{1}{\dopt}\sum_{i=\lambda_0+1}^\lambda k_i.
$$

For both terms, we prove that they are asymptotically bounded by $\sqrt{k}$.
%

First, using Equation~(\ref{eq:boundki}) and $\dopt> 2^{i_0+\lambda-1}$, we have:
\begin{align*}
\frac{1}{\dopt}\sum_{i=1}^{\lambda_0} k_i \leq
\frac{2}{2^{i_0+\lambda-1}}
\sum_{i=1}^{\lambda_0}\left(2^{i_0+i+2}\right)^2 \leq
\frac{2^{i_0+\lambda+3}}{2^{i_0+\lambda-1}}
\sum_{i=1}^{\lambda_0}2^{i_0+i+2} & \leq
2^6\,2^{i_0+\lambda_0+1}\\
& =
2^6\frac{M}{k^c} =
2^6 k^{\frac{1}{2-c}-c}
\end{align*}

Since $0\le c\le \frac{3}{2}$, then $\frac{1}{2-c}-c\leq \frac12$.
Consequently, 
$\frac{1}{\dopt}\sum_{i=1}^{\lambda_0} k_i = O(\sqrt{k})$.


Second, since $k_i\leq k$ for every $i\leq \lambda$, we have
$$\frac{1}{\dopt}\sum_{i=\lambda_0+1}^{\lambda} k_i \leq
\frac{1}{\dopt}(\lambda-\lambda_0)k \leq
\frac{2}{M}\bigl(c\log_2(k)+1\bigr)k \le 2\bigl(c\log_2(k)+1\bigr)k^{1-\frac{1}{2-c}}.$$

Define:
$$
    R(c,k)
    :=
    \frac{2(c\log_2(k)+1)\,
    k^{1-\frac{1}{2-c}}}{\sqrt{k}}
    =
    2c\log_2(k)\,
    k^{\frac12-\frac{1}{2-c}}.
$$
As $0\le c < 2$, we have $\frac12-\frac{1}{2-c}=-\frac{c}{2(2-c)}\leq -\frac{c}{4}$.
Consequently,
$$
    R(c,k)
    \leq
    2(c\log_2(k)+1)2^{-\frac{c}{4}\log_2(k)}.
$$
Let $x=c\log_2(k)$. Then, $R(c,k) \leq 2(x+1) 2^{-x/4}$. The function $x\mapsto (x+1)2^{-x/4}$, for $x\geq0$, attains its maximum
at $x=\frac{4}{\ln(2)}-1$. Hence,
$$
    R(c,k) \leq \frac{4\cdot 2^{1/4}}{e\ln(2)} < 3.
$$
Therefore, $\frac{1}{\dopt}\sum_{i=\lambda_0+1}^{\lambda}k_i=O(\sqrt{k})$.

For Scenario 2, we also have $\frac{1}{\dopt}\sum_{i=1}^\lambda k_i = O(\sqrt{k})$. As this holds for both scenarii, then the competitive ratio of \dflw\ is $O(\sqrt{k})$.
\end{proof}

\section{Conclusion}

We exhibited a deterministic strategy achieving a competitive ratio $O(\sqrt{k})$ for the \kctp\ on unit-weighted solid grid graphs, and proved that this bound is asymptotically tight. 
To the best of our knowledge, this is the first competitive ratio for the \kctp\ that lies strictly between the previously known regimes of $O(1)$ and $O(k)$. This raises the question of whether other intermediate regimes can be obtained.

In~\cite{BeBeChDa24}, a lower bound of $\Omega(\log k)$ was identified for arbitrarily-weighted outerplanar graphs. Whether any strategy can achieve a competitive ratio matching this lower bound remains an open question. Nevertheless, the existence of an optimal competitive ratio of $9$ for outerplanar graphs in the unit-weight setting suggests that this family of graphs is a promising candidate for exhibiting another sublinear competitive regime.

Moreover, it seems at first sight that our lower bound of $\Omega(\sqrt{k})$ can be extended to $d$-dimensional grid graphs, for $d\geq 2$, yielding a lower bound of $\Omega(k^{1-1/d})$. It would be interesting to investigate whether the \dflw\ strategy can also be generalized to higher dimensions.

\bibliography{ctp}

@article{BeWe15,
	title = {An optimal randomized online algorithm for the k-{Canadian} {Traveller} {Problem} on node-disjoint paths},
	volume = {30},
	number = {1},
	journal = {J. Comb. Optim.},
	author = {Bender, M. and Westphal, S.},
	year = {2015},
	pages = {87--96}
}

@article{We08,
	title = {A note on the $k$-{Canadian} {Traveller} {Problem}},
	volume = {106},
	number = {3},
	journal = {Inform. Proces. Lett.},
	author = {Westphal, S.},
	year = {2008},
	pages = {87--89}
}

@article{DeHuLiSa14,
	title = {Canadians {Should} {Travel} {Randomly}},
	journal = {Proc. of ICALP},
	author = {Demaine, E. D. and Huang, Y. and Liao, C. and Sadakane, K.},
	year = {2014},
	pages = {380--391}
}

@article{PaYa91,
	title = {Shortest paths without a map},
	author = {Papadimitriou, C. and Yannakakis, M.},
	journal = {Theor. Comput. Sci.},
	volume = {84},
	number = {1},
	pages = {127--150},
	year = {1991}
}

@article{XuHuSuZh09,
	title = {The {C}anadian traveller problem and its competitive analysis},
	volume = {18},
	number = {2},
	urldate = {2016-11-08},
	journal = {J. Comb. Optim.},
	author = {Xu, Y. and Hu, M. and Su, B. and Zhu, B. and Zhu, Z.},
	year = {2009},
	pages = {195--205}
}

@inproceedings{BaSc91,
  author    = {Bar{-}Noy, A. and Schieber, B.},
  title     = {The {C}anadian {T}raveller {P}roblem},
  booktitle = {Proc. of ACM/SIAM SODA},
  pages     = {261--270},
  year      = {1991}
}

@book{BoEl98,
  author    = {Borodin, A. and
               El{-}Yaniv, R.},
  title     = {Online computation and competitive analysis},
  publisher = {Cambridge Univ. Press},
  year      = {1998}
}

@article{AkSaAr16,
  author    = {V. Aksakalli and
               O. F. Sahin and
               I. Ari},
  title     = {An {AO\({}^{\mbox{*}}\)} Based Exact Algorithm for the Canadian Traveler
               Problem},
  journal   = {{INFORMS} Journal on Computing},
  volume    = {28},
  number    = {1},
  pages     = {96--111},
  year      = {2016}
}

@article{ChChWuWu15,
title = {{The $k$-Canadian Traveller Problem on Equal-Weight Graphs}},
author = {Chan, H. and
Chang, J. and 
Wu, H. and Wu, T.},
journal = {Proc. of WCMCT},
pages = {135--137},
year = {2015}
}

@article{LiHu14,
  author    = {C. Liao and
               Y. Huang},
  title     = {The Covering Canadian Traveller Problem},
  journal   = {Theor. Comput. Sci.},
  volume    = {530},
  pages     = {80--88},
  year      = {2014}
}

@inproceedings{HaXe23,
  author       = {N. Hahn and
                  M. Xefteris},
  title        = {The Covering Canadian Traveller Problem Revisited},
  booktitle    = {Procs. of {MFCS}},
  series       = {LIPIcs},
  volume       = {272},
  pages        = {53:1--53:12},
  year         = {2023}
}

@article{BeSa23,
  author       = {P. Berg{\'{e}} and
                  L. Sala{\"{u}}n},
  title        = {The influence of maximum (\emph{s}, \emph{t})-cuts on the competitiveness
                  of deterministic strategies for the {Canadian Traveller Problem}},
  journal      = {Theor. Comput. Sci.},
  volume       = {941},
  pages        = {221--240},
  year         = {2023}
}

@article{fried2013complexity,
  title={Complexity of Canadian traveler problem variants},
  author={D. Fried and S. E. Shimony and A. Benbassat and C. Wenner},
  journal={Theor. Comput. Sci.},
  volume={487},
  pages={1--16},
  year={2013},
  publisher={Elsevier}
}

@article{AlYiAk21,
  author       = {A. F. Alkaya and
                  S. Yildirim and
                  V. Aksakalli},
  title        = {Heuristics for the {Canadian Traveler Problem} with neutralizations},
  journal      = {Comput. Ind. Eng.},
  volume       = {159},
  pages        = {107488},
  year         = {2021}
}

@inproceedings{EyKeHe10,
  author       = {P. Eyerich and
                  T. Keller and
                  M. Helmert},
  title        = {High-Quality Policies for the {Canadian Traveler's Problem}},
  booktitle    = {Proces. of {AAAI}},
  pages        = {51--58},
  publisher    = {{AAAI} Press},
  year         = {2010}
}

@article{BeBeChDa24,
  author       = {L. Beaudou and
                  P. Berg{\'{e}} and
                  V. Chernyshev and
                  A. Dailly and
                  Y. Gerard and
                  A. Lagoutte and
                  V. Limouzy and
                  L. Pastor},
  title        = {The {Canadian Traveller Problem} on unit-weighted and arbitrarily weighted outerplanar graphs},
  journal      = {Theor. Comput. Sci.},
  volume       = {1082},
  pages        = {116096},
  year         = {2026},
  doi          = {10.1016/J.TCS.2026.116096}
}

@inproceedings{UmLe97,
  author       = {C. Umans and
                  W. J. Lenhart},
  title        = {Hamiltonian Cycles in Solid Grid Graphs},
  booktitle    = {Procs. of {FOCS}},
  pages        = {496--505},
  publisher    = {{IEEE} Computer Society},
  year         = {1997}
}

@inproceedings{BaEsXe22,
  author       = {E Bampis and
                  B. Escoffier and
                  M. Xefteris},
  editor       = {Parinya Chalermsook and
                  Bundit Laekhanukit},
  title        = {{Canadian Traveller Problem} with Predictions},
  booktitle    = {Procs. of {WAOA}},
  series       = {Lecture Notes in Computer Science},
  volume       = {13538},
  pages        = {116--133},
  publisher    = {Springer},
  year         = {2022},
  doi          = {10.1007/978-3-031-18367-6\_6}
}

\appendix

\section{Proofs of Section~\ref{sec:prelim}}

\begin{proof}[Proof of Theorem~\ref{thm:lower-bound-weight}]
Let $A$ be a deterministic strategy for the \kctp. We first prove the result for unit-weighted grid graphs.

For two integers $k,\ell\geq 1$, let $H_{k,\ell}$ be the graph defined as follows:
\begin{itemize}
    \item $P=v_1\ldots v_{2k+1}$ and $P'=v'_1\ldots v'_{2k+1}$ are two paths on $2k+1$ vertices;
    \item for each $0\leq i\leq k$, let
    $Q_i=u_0^{(i)}\ldots u_\ell^{(i)}$ be a path of length $\ell$ such that
    $u_0^{(i)}=v_{2i+1}$ and $u_\ell^{(i)}=v'_{2i+1}$.
\end{itemize}
We set $s=v_1$ and $t=v'_{2k+1}$ (see Figure~\ref{fig:LowerBoundGrid}). All edges of $H_{k,\ell}$ have unit weight.

Consider the following set of $k+1$ edges: $F = \{u_{\ell-1}^{(i)}u_\ell^{(i)}:0\leq i\leq k\}$. It contains all horizontal edges incident to the rightmost vertical line. Set $F$ is an $(s,t)$-separator, hence the traveller must traverse at least one edge of set $F$ in order to reach $t$. We will consider blockages only inside set $F$.
More precisely, we fix $E_*$ as a set made up of $k$ among the $k+1$ edges of $F$, such that it contains the $k$ first edges which strategy $A$ aims at traversing to attain the rightmost vertical line. Concretely, strategy $A$ first tries to make the traveller pass through some $(u_{\ell-1}^{(i_1)},u_\ell^{(i_1)})$ but fail as we fix $(u_{\ell-1}^{(i_1)},u_\ell^{(i_1)}) \in E_*$. So, strategy $A$ tries to pass through another $(u_{\ell-1}^{(i_2)},u_\ell^{(i_2)})$ but the traveller is blocked as we fix $(u_{\ell-1}^{(i_2)},u_\ell^{(i_2)}) \in E_*$. Finally, the traveller reaches the rightmost vertical line by passing through some $(u_{\ell-1}^{(i_{k+1})},u_\ell^{(i_{k+1})})$ which is open.



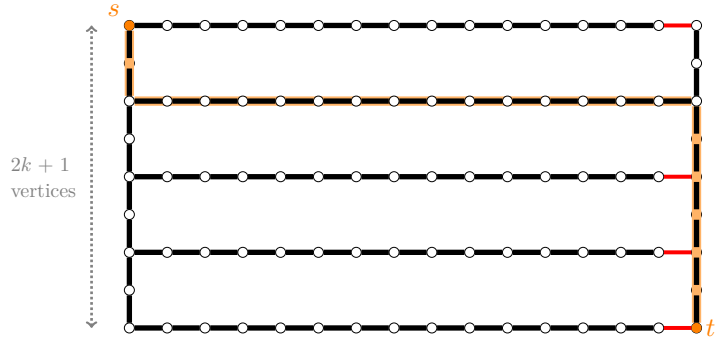
\begin{figure}[ht]
	\centering
	\scalebox{0.5}{\begin{tikzpicture}

\def\k{4}      
\def\ell{15}   

\pgfmathtruncatemacro{\twok}{2*\k}
\pgfmathtruncatemacro{\twokpo}{2*\k+1}
\pgfmathtruncatemacro{\ellmo}{\ell-1}

\draw[<->, dotted, line width = 2pt, color = gray] (-1,0) to (-1,\twok);
\node[color = gray, scale = 1.5,text width = 1.5cm] at (-2,\k) {$2k+1$ vertices};


\foreach \i in {0,...,\twok} {
    \node[vertex] (v-\i)  at (0,\i) {};
    \node[vertex] (vp-\i) at (\ell,\i) {};
}

\draw[line width = 7pt, color=white!40!orange] (v-\twok) -- (v-6) -- (vp-6) -- (vp-0);

\foreach \i in {1,...,\twok} {
    \pgfmathtruncatemacro{\imo}{\i-1}
    \draw[free] (v-\imo) -- (v-\i);
    \draw[free] (vp-\imo) -- (vp-\i);
}

\foreach \i in {0,...,\k} {

    \pgfmathtruncatemacro{\y}{2*\i}

    \foreach \j in {1,...,\ellmo} {
        \node[vertex] (u-\i-\j) at (\j,\y) {};
    }

    \draw[free] (v-\y) -- (u-\i-1);

    \foreach \j in {2,...,\ellmo} {
        \pgfmathtruncatemacro{\jmo}{\j-1}
        \draw[free] (u-\i-\jmo) -- (u-\i-\j);
    }


}
\foreach \i in {0,...,2} {
    \pgfmathtruncatemacro{\y}{2*\i}
    \draw[blocked] (u-\i-\ellmo) -- (vp-\y);
}

\draw[blocked] (u-4-\ellmo) -- (vp-8);
\draw[free] (u-3-\ellmo) -- (vp-6);


\node[terminal] (s) at (v-\twok) {};
\draw (s) node[above left,scale=2,color = orange] {$s$};

\node[terminal] (t) at (vp-0) {};
\draw (t) node[right,scale=2,color = orange] {$t$};

\end{tikzpicture}}
	\caption{An example of grid graph on which any strategy achieves at least ratio $2k-1$.}\label{fig:LowerBoundGrid}
	\end{figure}

Said differently, the $k$ first paths $Q_i$ traversed by the traveller are unsuccessful. These paths are traversed in both directions as the traveller is blocked on their extremity. Moreover, the traveller traverses path $Q_{i_{k+1}}$ in one direction to reach the rightmost vertical line. As a consequence, at least $2k(\ell-1)$ edges are traversed horizontally by the traveller in this configuration $E_*$. 
Vertically, the traveller traverses at least $2k$ edges from $s$ to $t$.
Hence, denoting by $L_A(k,\ell)$ the total length of the walk produced by $A$, we have:
$$
L_A(k,\ell)\geq 2k(\ell-1)+2k = 2k\ell.
$$

Every shortest $(s,t)$-path in $H_{k,\ell}$ has length $2k+\ell$, and one is open as $(u_{\ell-1}^{(i_{k+1})},u_\ell^{(i_{k+1})})$ is not blocked (see the orange path in Figure~\ref{fig:LowerBoundGrid}). Therefore,
$$
\dopt(k,\ell)=2k+\ell.
$$
For the competitive ratio, it follows that:
$$
\frac{L_A(k,\ell)}{\dopt(k,\ell)}
\geq
\frac{2k\ell}{2k+\ell}.
$$
Fixing $\ell=(2k)^2$ yields:
$$
\frac{2k\ell}{2k+\ell}
=
\frac{8k^3}{4k^2+2k}
=
\frac{8k^3 + 4k^2 - 4k^2}{4k^2+2k}
\ge 
2k-1.
$$
Hence, the competitive ratio of any deterministic strategy $A$ on unit-weighted grid graphs is at least $2k-1$.

Observe that $H_{k,\ell}$ is indeed a grid graph, but it is not a solid grid graph, as it contains non-unit bounded faces. We now adapt the construction to weighted solid grid graphs.

Let $H'_{k,\ell}$ be obtained from $H_{k,\ell}$ by adding all the missing vertices and edges between consecutive paths $Q_i$ and $Q_{i+1}$, so that $H'_{k,\ell}$ is a rectangular grid, and therefore a solid grid graph. We assign weight $1$ to every edge of $H_{k,\ell}$ and weight
$W=(k\ell)^2$ to every edge of $E[H'_{k,\ell}]\setminus E[H_{k,\ell}]$.

We use the same ``blocking adversary'' as above on the edges of $H_{k,\ell}$. If we assume that $A$ never makes the traveller traverse an edge of weight $M$, his exploration is confined to $H_{k,\ell}$, and the previous reasoning gives again a competitive ratio of at least $2k-1$.

Suppose now that $A$ traverses an edge of weight $M$. Its total cost is then at least $M$. On the other hand, the optimal path described above uses only edges of $H_{k,\ell}$ and therefore still has cost $2k+\ell$. Thus,
$$
\frac{L_A(k,\ell)}{\dopt(k,\ell)}
\geq
\frac{(k\ell)^2}{2k+\ell}.
$$
For $\ell=(2k)^2$, this quantity is larger than $2k-1$. Consequently, in both cases, the competitive ratio of $A$ is at least $2k-1$.
Therefore, the lower bound $2k-1$ holds both for unit-weighted grid graphs and for weighted solid grid graphs.
\end{proof}

\section{Proofs of Section~\ref{sec-double}} \label{app:double}

\begin{proof}[Proof of Lemma~\ref{lemma:incr-k}]
Let $R=2^{i_0+i-1}$. Thus, $V_{i-1}$ is induced by the vertices at
$\ell_\infty$-distance at most $R$ from $s$, while $V_i$ is induced by
the vertices at $\ell_\infty$-distance at most $2R$ from $s$.

Suppose, towards a contradiction, that $k_i-k_{i-1}<R$. For every
$r\in\{R+1,\ldots,2R\}$, let
$$
    \Gamma_r=\{x\in\mathbb{Z}^2:d_{\infty}(x,s)=r\}.
$$
There are $R$ such square boundaries, and every grid edge is contained
in at most one of them. Since we assumed that there are fewer than $R$ blocked edges in
$E(G[V_i])\setminus E(G[V_{i-1}])$, there exists
$r_0\in\{R+1,\ldots,2R\}$ such that no edge of $G$ contained in
$\Gamma_{r_0}$ is blocked. Fix such an $r_0$ and write
$\Gamma=\Gamma_{r_0}$.

Let $Q$ be a shortest $(s,t)$-path contained in $G[V_1]$. Such a path exists
as $N=d_{\emptyset}(s,t)$ is smaller than $2^{i_0+1}$, the radius of $V_1$.
Since $i\geq 2$, we have $V_1\subseteq V_{i-1}$, and therefore
$Q\subseteq V_{i-1}$. In particular, $Q$ is contained strictly inside
the square bounded by $\Gamma$.

As \flw\ fails on $(G[V_i],s,t)$ (because $i<\lambda$) and \flw\ is correct, there is no open $(s,t)$-path contained in $G[V_i]$. Let $P$ be an open $(s,t)$-path in $G$. We may assume that $P$ is simple. Necessarily, $P$ contains a vertex outside $V_i$.

Let $z\in V(P)\setminus V_i$. Following $P$ from $s$ to $t$, let $a$ be the last vertex of $P\cap Q$ preceding $z$ (it exists as $s$ is a candidate), and let $b$ be the first vertex of $P\cap Q$ following $z$ (it exists as $t$ is a candidate). 
By their choice,
$$
    V(P[a,b])\cap V(Q)=\{a,b\}.
$$
Consequently, $C=P[a,b]\cup Q[a,b]$ is a simple cycle. Moreover, $P[a,b] \subsetneq Q$ contains the vertex $z\notin V_i$. Figure~\ref{fig:doubling-cycle} illustrates cycle $C$ with both sections $P[a,b]$ (in red) and $Q[a,b]$ (in blue).

Among all choices of an open $(s,t)$-path $P$, a vertex $z\in V(P)\setminus V_i$, and the corresponding vertices $a,b$ defined above, select a combination for which the region enclosed by $C$ has minimum area. Let $\operatorname{Int}(C)$ contain all vertices of $G$ surrounded by $C$ in the grid embedding. We write $\overline{\operatorname{Int}}(C)=\operatorname{Int}(C)\cup C$
for the closed region bounded by $C$. As $G$ is a solid grid, we can show that every lattice point of $\overline{\operatorname{Int}}(C)$ is a vertex of $G$.

\begin{claim}\label{claim:solidity}
For every $x\in\mathbb{Z}^2\cap\overline{\operatorname{Int}}(C)$, we have $x\in V(G)$.
\end{claim}

\begin{claimproof}
Suppose, for the sake of contradiction, that there exists $x\in\mathbb{Z}^2\cap\overline{\operatorname{Int}}(C)$ such that $x\notin V(G)$. Since every vertex of $C$ is a vertex of $G$, we necessarily have $x\in\operatorname{Int}(C)$.

Consider the face $f$ of $G$ containing $x$. Since $x$ lies in the interior of the simple cycle $C$, the face $f$ is an inner face of $G$. Moreover, since $G$ is an induced subgraph of the square grid and $x\notin V(G)$, none of the four grid edges incident with $x$ belongs to $G$. Consequently, the four unit squares incident with $x$ are contained in the same face $f$. Hence,
$$
    \operatorname{area}(f)\geq 4,
$$
contradicting the assumption that $G$ is solid, {\em i.e.} every inner face has unit area.
\end{claimproof}

\begin{figure}[t]
    \centering
    \begin{tikzpicture}[
        scale=1,
        every node/.style={font=\small},
        vertex/.style={circle, fill=black, inner sep=1.6pt}
    ]

        \draw[ultra thick]
            (-3,-3) rectangle (3,3);

        \coordinate (s) at (-2,0);
        \coordinate (a) at (-1.2,0.6);
        \coordinate (b) at (1.2,0.4);
        \coordinate (t) at (2,-0.2);

\fill[red!15]
    (a)
    to[out=-100,in=0] (-1.5,-3)
    to[out=180,in=-70] (-3.7,-1.5)
    to[out=110,in=180] (-2.5,2.8)
    to[out=0,in=160] (-1.5,1.8)
    to[out=-20,in=-120] (-0.3,3)
    to[out=60,in=160] (3.8,2.2)
    to[out=-20,in=60] (3.7,-2.0)
    to[out=-120,in=-80] (b)

    .. controls (0.7,1.2) and (0.2,-0.3) .. (a)
    -- cycle;

         
        \draw[thick, dashed]
            (-3,-3) rectangle (3,3);

          \node at (3.4,3.4) {$\Gamma$};

        \draw[thick, blue]
            (s)
           .. controls (-1.7,0.8) and (-1.2,0.9) .. (a)
            .. controls (0.2,-0.3) and (0.7,1.2) .. (b)
            .. controls (1.3, 0) and (1.6,0) .. (t);

\draw[thick, red]
    (a)
    to[out=-100,in=0] (-1.5,-3)
    to[out=180,in=-70] (-3.7,-1.5)

    to[out=110,in=180] (-2.5,2.8)
    to[out=0,in=160] (-1.5,1.8)

    to[out=-20,in=-120] (-0.3,3)

    to[out=60,in=160] (3.8,2.2)
    to[out=-20,in=60] (3.7,-2.0)

    to[out=-120,in=-80] (b);

        \node[vertex, label=below:$s$] at (s) {};
        \node[vertex, label=below:$t$] at (t) {};
        \node[vertex, label=below left:$a$] at (a) {};
        \node[vertex, label=right:$b$] at (b) {};

        \node[blue] at (0,-0.15) {$Q$};
        \node[red] at (0.2,4) {$P[a,b]$};

    \end{tikzpicture}
    \caption{The cycle formed by $P[a,b]\cup Q[a,b]$. The dashed arcs of $\Gamma$ represent the possible shortcuts.}
    \label{fig:doubling-cycle}
\end{figure}
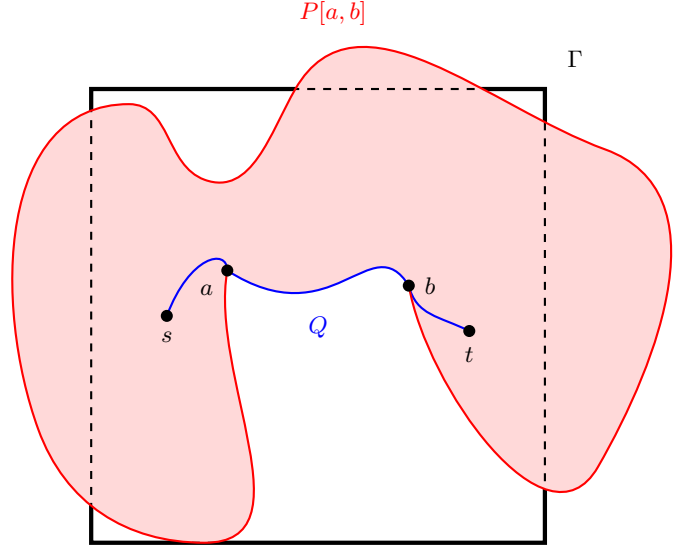

We now show that a portion of $\Gamma$ gives an open shortcut for $P$.
Since $a,b\in Q\subseteq V_{i-1}$ and $P[a,b]$ contains a vertex
outside $V_i$, the path $P[a,b]$ must leave the square bounded by
$\Gamma$. Let $x$ be a vertex at which $P[a,b]$ leaves this square;
that is, $x\in\Gamma$ and the successor $y$ of $x$ along $P[a,b]$ lies
strictly outside $\Gamma$.

Let $x_1$ and $x_1'$ be the two neighbors of $x$ along $\Gamma$.

\begin{claim}
At least one of $x_1$ and $x_1'$ belongs to
$\overline{\operatorname{Int}}(C)$.
\end{claim}

\begin{claimproof}
Indeed, in a sufficiently small neighborhood of $x$, the two edges of
$C$ incident with $x$ divide the neighborhood into an interior and an
exterior sector. One of these two edges is the edge of $P[a,b]$ by
which $P$ reaches $x$, and the other is $(x,y)$, the edge by which $P$ leaves
the square bounded by $\Gamma$. Hence one of the two directions of
$\Gamma$ incident with $x$ lies on the interior side of the Jordan
curve $C$. Therefore at least one of $x_1,x_1'$ belongs to
$\overline{\operatorname{Int}}(C)$.
\end{claimproof}

Assume w.l.o.g. that
$x_1\in\overline{\operatorname{Int}}(C)$, and traverse $\Gamma$ from
$x=x_0$ in the direction of $x_1$. Let
$$
    A=x_0x_1\cdots x_p
$$
be the maximal subpath of $\Gamma$ all of whose vertices belong to $\overline{\operatorname{Int}}(C)$. Note that, by Claim~\ref{claim:solidity}, $A$ is a path in $G$.

The path $A$ cannot traverse the whole of $\Gamma$. Indeed, $P[a,b]$ eventually returns to the interior of $\Gamma$, since $b\in Q\subseteq V_{i-1}$. At such a re-entry point, the same local Jordan-curve argument, with the roles of the interior and exterior reversed, shows that one of the two directions of $\Gamma$ lies on the exterior side of $C$. Thus, when $\Gamma$ is traversed from $x$ in the interior direction chosen above, one must eventually leave $\overline{\operatorname{Int}}(C)$.

Let $x_{p+1}$ be the vertex of $\Gamma$ immediately following $x_p$. By maximality, $x_{p+1}\notin\overline{\operatorname{Int}}(C)$. Since the edge $x_px_{p+1}$ joins the closed interior of $C$ to its exterior, and $x_{p+1}\notin C$; hence $x_p\in C$. 
Thus $A$ is an $(x_0,x_p)$-path whose endpoints lie on $C$ and which is otherwise contained in $\operatorname{Int}(C)$.
Furthermore, since $Q$ lies strictly inside $\Gamma$, neither $x_0$ nor
$x_p$ belongs to $Q$. Hence both belong to $P[a,b]$.

Since $A\subseteq\Gamma$ and $\Gamma$ contains no
blocked edge, $A$ is an open $(x_0,x_p)$-path. Replace the $(x_0,x_p)$-subpath of $P$ lying on the boundary of the region cut off by $A$ by the path $A$. Since both the replaced part and $A$ have the same endpoints and $A$ is contained in $\overline{\operatorname{Int}}(C)$, this yields an open $(s,t)$-path $P'$. 

Moreover, replacing $P$ by $P'$ removes a nonempty region from the region bounded by $C$. In particular, every cycle obtained from the modified portion of $P'$ and the corresponding portion of $Q$ is contained in $\overline{\operatorname{Int}}(C)$, and the cycle containing the modified portion has strictly smaller enclosed area than $C$.

Finally, since $P'$ is an open $(s,t)$-path and \flw\ fails on $G[V_i]$, correctness of \flw\ implies that $P'$ is not contained in $V_i$. Choose $z'\in V(P')\setminus V_i$, and let $a'$ and $b'$ be, respectively, the last and first vertices of $P'\cap Q$ before and after $z'$. Then
$$
    C'=P'[a',b']\cup Q[a',b']
$$
is a simple cycle of the same type as $C$. By the shortcut construction, $C'$ is included into $\overline{\operatorname{Int}}(C)$ and has a smaller enclosed area. This contradicts the minimal choice of $C$.

Therefore $k_i-k_{i-1}\geq R$. As $R=2^{i_0+i-1}$, we conclude that $k_i\geq k_{i-1}+2^{i_0+i-1}$.
\end{proof}

\end{document}